\documentclass[notitlepage,superscriptaddress,nofootinbib,tightenlines]{revtex4-1}
\usepackage{amsmath,verbatim,latexsym,amssymb,graphicx,indentfirst,mathrsfs,mathtools,amsthm,bbm,bm,hyperref,url,cancel,subcaption}
\usepackage[font=small,labelfont=bf,format=plain,justification=raggedright,singlelinecheck=false]{caption}

\usepackage{verbatim,indentfirst}
\usepackage[title,titletoc]{appendix}

\makeatletter
\def\p@subsection{}
\makeatother
\makeatletter
\def\p@subsubsection{}
\makeatother

\newtheorem{theorem}{Theorem}

\newtheorem{definition}{Definition}
\newtheorem{example}{Example}
\newtheorem{proposition}{Proposition}

\makeatletter
\newcommand\footnoteref[1]{\protected@xdef\@thefnmark{\ref{#1}}\@footnotemark}
\makeatother

\makeatletter
\renewcommand{\p@subsection}{}
\renewcommand{\p@subsubsection}{}
\makeatother

\newcommand*{\ket}[1]{\lvert #1 \rangle}

\newcommand*{\ketbra}[2]{\lvert #1 \rangle\!\langle #2 \rvert}

\newcommand{\kB}{k_\mathrm{B}}

\DeclareMathOperator{\supp}{supp}

\newcommand{\tr}{{\rm Tr}}
\def\id{\mathbbm{1}}
\def\Dh{D_{\rm H}}
\def\Wext{W_{\rm gain}^{\varepsilon}}
\def\Wform{W_{\rm cost}^{\varepsilon}}

\newcommand{\caphead}[1]{{\bf #1}}

\usepackage{color}
\usepackage[dvipsnames]{xcolor}
\usepackage[normalem]{ulem}

\begin{document}

\title{Beyond heat baths II: Framework for generalized thermodynamic resource theories}

\author{Nicole~Yunger~Halpern\footnote{E-mail: nicoleyh@caltech.edu}}
\affiliation{Institute for Quantum Information and Matter, Caltech, Pasadena, CA 91125, USA}
\affiliation{Perimeter Institute for Theoretical Physics, 31 Caroline Street North, Waterloo, Ontario Canada N2L 2Y5}

\date{\today}

%
%
%
%
\begin{abstract}

Thermodynamics, which describes vast systems, has been reconciled with small scales, relevant to single-molecule experiments, in \emph{resource theories}. Resource theories have been used to model exchanges of energy and information. Recently, particle exchanges were modeled; and an umbrella family of thermodynamic resource theories was proposed to model diverse baths, interactions, and free energies. This paper motivates and details the family's structure and prospective applications. How to model electrochemical, gravitational, magnetic, and other thermodynamic systems is explained. Szil\'{a}rd's engine and Landauer's Principle are generalized, as resourcefulness is shown to be convertible not only between information and gravitational energy, but also among diverse degrees of freedom. Extensive variables are associated with quantum operators that might fail to commute, introducing extra nonclassicality into thermodynamic resource theories. An early version of this paper partially motivated the later development of noncommutative thermalization. This generalization expands the theories' potential for modeling realistic systems with which small-scale statistical mechanics might be tested experimentally.

\end{abstract}

{\let\newpage\relax\maketitle}


%
%
%
%
\section{Introduction}

Thermodynamics models diverse systems, from gases and magnets to chemical reactions, soap bubbles, and electrochemical batteries. Thermodynamic systems contain on the order of $10^{24}$ particles. But thermodynamic concepts such as heat, work, and entropy are relevant to small systems such as molecular motors and ratchets~\cite{LacosteLM08,SerreliLKL07}, the unfolding of single DNA and RNA molecules~\cite{BustamanteLP05,CollinRJSTB05,LiphardtDSTB02,AlemanyR10}, and nanoscale walkers~\cite{ChengSHEL12}. How can small scales, increasingly controllable in experiments, be reconciled with thermodynamics?
The \emph{resource-theory framework} describes small-scale exchanges of energy~\cite{Janzing00,FundLimits2,BrandaoHORS13,BrandaoHNOW14}; information~\cite{HHOLong,GourMNSYH13,FaistDOR12}; and, in recent work, particles~\cite{YungerHalpernR14}. These successes call for a generalization of thermodynamic resource theories to diverse realistic systems. This paper details the generalization proposed in~\cite{YungerHalpernR14}.

Thermodynamic resource theories fall under the umbrella of \emph{one-shot statistical mechanics}~\cite{Dahlsten13}, an application of \emph{one-shot information theory} that generalizes Shannon theory. Shannon theory quantifies the efficiencies with which protocols, such as data compression, can be performed as the number $n$ of trials approaches infinity: $n \to \infty$. Infinitely many trials are never performed in reality. The efficiencies of finitely many trials, and of faulty protocols, are quantified with one-shot information theory~\cite{RennerThesis}. 

A \emph{resource theory} quantifies the value attributable to quantum states by an agent who can perform only certain operations easily~\cite{HorodeckiO13,CoeckeFS14}. Examples include the resource theories for pure bipartite entanglement (which I will call \emph{entanglement theory}, for short)~\cite{HorodeckiHHH09}, asymmetry~\cite{BartlettRS07,marvian_theory_2013,BartlettRST06}, and quantum computation~\cite{VeitchMGE14}. Each resource theory is cast in terms of an agent who can process states in certain ways---via \emph{free operations}---at zero cost. The free operations in the entanglement theory, for example, are local operations and classical communications (LOCC). Free operations include the creation of \emph{free states}, such as product states in the entanglement theory. States that are not free, e.g., entangled states, are \emph{resources}. Resources have value because they, with free operations, can simulate nonfree operations. Combined with LOCC, a Bell pair can simulate quantum teleportation.

In the presence of a temperature-$T$ heat bath, nonequilibrium states have value because work can be extracted from them and work is needed to create them. Resource theories have been used to quantify the work extractable from, and the work cost of, single copies of states (e.g.,~\cite{Janzing00,BrandaoHORS13,FundLimits2,SkrzypczykSP13Extract}). As the number $n$ of copies approaches infinity, the work extractable, and the work cost, per copy converges to a function of the Helmholtz free energy $F := E - TS$. This convergence motivates the labeling of these resource theories as \emph{Helmholtz theories}~\cite{YungerHalpernR14}. Each Helmholtz theory is distinguished by the value $\beta$ of the  bath's inverse temperature. A Helmholtz theory in which all Hamiltonians are trivial ($H = 0$) has been portrayed as a resource theory for information~\cite{HHOShort,HHOLong}, called the \emph{resource theory of nonuniformity}~\cite{GourMNSYH13}. As justified below, the theory will be called the \emph{resource theory for entropy}, or the \emph{entropy theory}, after the thermodynamic potential that quantifies resourcefulness.

Helmholtz theories were generalized recently in~\cite{YungerHalpernR14}. Just as a system's Helmholtz free energy can be transformed into work in the presence of a heat bath, so can a system's grand potential in the presence of a heat-and-particle bath. So can the Gibbs free energy and other thermodynamic potentials, in other settings. Systems exchange not only heat and information, but also particles, volume, and angular momentum~\cite{Callen85,JaynesII57}. The success of Helmholtz theories invites us to generalize resource theories to the rest of thermodynamics. A generalization was proposed and was illustrated with grand-potential theories, which model exchanges of energy and particles, in~\cite{YungerHalpernR14}. 

This paper motivates and details the umbrella family of thermodynamic resource theories. Section~\ref{section:Background} introduces notation and background. Generalized thermodynamic resource theories are discussed in Sec.~\ref{section:DefineTheories}. Free \emph{equilibrating operations}, which conserve quantities such as heat, particle number, and volume, are defined in Sec.~\ref{section:FreeOps}. 
From free operations, the equilibrium form of free-state density operators is derived in Sec.~\ref{section:FreeStates}. 

Section~\ref{section:Quasiorder} demonstrates the equivalence between the quasiorder induced by free operations and $d$-majorization relative to equilibrium states. In Sec.~\ref{section:Battery}, the storage of work in diverse forms, such as chemical energy, gravitational energy, and electrical energy, is explored. The work $\Wext(R)$ extractable from, and the work $\Wform(R)$ required to create, one copy of a state $R$ with a faulty protocol is quantified in Sec.~\ref{section:OneShotWork}. 
$\Wext(R)$ and $\Wform(R)$ converge, in the asymptotic limit introduced in Sec.~\ref{section:ThermLimit}, to a difference between free energies. 

Open problems are discussed in Section~\ref{section:Discussion}. 
Opportunities include noncommuting operators
that model extensive variables.
Since this paper was first released, in 2014,
the theory of such noncommutation has been developed~\cite{Lostaglio_17_Thermodynamic,Guryanova_16_Thermodynamics,YH_16_Microcanonical}.
Further works (e.g.,~\cite{Ito_16_Optimal,NathBera_17_Thermodynamics,MurPetit_17_Generalised})
have built on these foundations.
The present paper provided one of the original motivations.
Related motivations were developed independently in~\cite{Lostaglio_14_Masters}.

This paper's contributions are largely conceptual.
The general framework for thermodynamic resource theories
subsumes already-defined theories,
is illustrated with new theories,
and dictates how to construct more theories.
The framework's construction highlights facets of thermodynamics
that have remained unrecognized in
pre-existing thermodynamic resource theories.
For example, Helmholtz theories are cast in terms of
the energy representation of conventional thermodynamics.
The resource theory of nonuniformity
is cast in terms of
the entropy representation 
(Sections~\ref{section:FundRelation} and~\ref{section:Massieu}).
Misconceptions about thermodynamic resource theories are clarified.
For example, the entropy theory is often cast as
a Helmholtz theory whose $\beta = \infty$ or 0.
But the entropy theory models systems
whose energies, volumes, and particle numbers remain constant.
Helmholtz theories model systems that exchange heat.
This discrepancy differentiates the entropy theory's free operations
from Helmholtz theories' (Sec.~\ref{section:FreeOps}).
I also identify and formalize hidden assumptions 
in thermodynamic resource theories.
For example, each system modeled in any thermodynamic resource theory 
implicitly has physical degrees of freedom
(particle number, angular momentum, etc.)
not modeled explicitly in Helmholtz theories.
I show how to model these degrees of freedom.
The implicit preservation of these ``hidden'' observables,
I concretize in the fixed-eigensubspace condition in Sec.~\ref{section:Fixed_eigen}.
Additionally, I introduce batteries whose resourcefulness manifests
in degrees of freedom other than energy
(e.g., in a high particle concentration).
(Angular momentum was identified as such a degree of freedom independently
in~\cite{VaccaroB11,BarnettV13}.)
Also the open questions are original.
They have motivated, for example, work on 
thermodynamic exchanges of noncommuting observables
that was published after the initial version of this paper was released.
Technical results' proofs are generalized straightforwardly 
from~\cite{YungerHalpernR14}.

This generalization of thermodynamic resource theories opens diverse, realistic systems to modeling by one-shot statistical mechanics. The generalization is intended to facilitate experimental tests of one-shot theory.

%
%
%
%
\section{Background and notation}
\label{section:Background}

Thermodynamics will guide the generalization of thermodynamic resource theories. I will review intensive and extensive variables, the Fundamental Relation of Thermodynamics, the energy and entropy representations, thermodynamic potentials, natural variables, and Massieu functions. Readers might know much of this material.
However, notation (adopted from Callen~\cite{Callen85}) and topics such as Massieu functions might be unfamiliar.

%
%
\subsection{Intensive and extensive variables}

Every thermodynamic system $\mathcal{S}$ has thermodynamic variables---properties, some controlled by the experimenter, that characterize $\mathcal{S}$ as a whole. Examples appear in Table~\ref{table:Variables}.

\emph{Intensive variables} remain constant as the system's size changes. Examples include the temperature $T$, the pressure $p$, and the chemical potential $\mu^\alpha_i$ associated with phase $\alpha$ of chemical species $i$. 
If $\mathcal{S}$ is polarizable, magnetizable, subject to mechanical forces, etc., then the external electric field $\mathbf{E}$, the external magnetic field $\mathbf{B}$, the stress, etc. serve as intensive variables~\cite{Alberty01}. I will distinguish between \emph{energy intensive variables} $p_i$ and \emph{entropy intensive variables} $F_i$ in Sec.~\ref{section:FundRelation}. 

\emph{Extensive variables} scale with the system's size. Examples include the energy $E$, the thermodynamic entropy $S$, the volume $V$, the number $N^\alpha_i$ of species-$i$ particles in phase $\alpha$, polarization, magnetization, and mechanical strain~\cite{Alberty01}. I will denote $E$ by $X_0$ and will denote the $i^{\rm th}$ of the extensive variables other than $E$ and $S$ by $X_i$.

%
%
\begin{table*}[t]
\begin{center}
\begin{tabular}{|c|c|c|} 
   \hline
        Type of work
   &   Intensive variables
   &   Extensive variables  \\
   \hline \hline 
        Mechanical
   &   $-p$
   &   $V$  \\
         (pressure-vol.)   
   &   
   &      \\
   \hline 
        Chemical
   &   $\mu^\alpha_i$  
   &   $N^\alpha_i$   \\
   \hline
        Gravitational   
   &   $\phi  =  gh$
   &   $m$   \\
   \hline 
        Electrochemical   
   &   $\bar{\mu}^\alpha_i$ 
   &   $N^\alpha_i$   \\
   \hline
         Magnetic   
   &   $\mathbf{B}$
   &   $\mathbf{m}$   \\
   \hline
        Electrical
   &   $\mathbf{E}$
   &   $\mathbf{p}$   \\
   \hline 
   \hline 
\end{tabular}
\caption{\caphead{Example energy intensive variables, extensive variables, and types of thermodynamic work:} Pressure-volume work involves the negative $-p$ of pressure, as well as the volume $V$. Chemical work involves the standard chemical potential $\mu^\alpha_i$ of phase $\alpha$ of species $i$, as well as the number $N^\alpha_i$ of phase-$\alpha$ species-$i$ particles. Gravitational work involves the gravitational potential $\phi = gh$ and the mass $m$ (alternatively, as explained in Sec.~\ref{section:Battery}, the gravitational chemical potential $\tilde{\mu}_i^\alpha$ and the number $N^\alpha_i$ of species-$i$ particles at height $\alpha$). Electrochemical work involves the electrochemical potential $\bar{\mu}^\alpha_i$; magnetic work, the magnetic moment $\mathbf{m}$ and the external magnetic field $\mathbf{B}$; and electrical work, the polarization $\mathbf{p}$ and the external electric field $\mathbf{E}$.}
\label{table:Variables}
\end{center} 
\end{table*}

%
%
\subsection{Fundamental Relation; energy and entropy representations}
\label{section:FundRelation}

The Fundamental Relation of Thermodynamics ``contains \emph{all} thermodynamic information about the system''~\cite{Callen85}. 
One can calculate from the relation, for example, conditions under which $\mathcal{S}$ is in equilibrium. The relation's \emph{energy representation} is
\begin{equation}   \label{eq:FundRelEnergy}
   E  =  E(S, X_1, X_2, \ldots, X_k);
\end{equation}
its equivalent \emph{entropy representation},
\begin{equation}   \label{eq:FundRelEntropy}
   S  =  S(X_0, X_1, X_2, \ldots, X_k).
\end{equation}
If $E$ depends only on $V$, on the $N^\alpha_i$, and on $S$, $\mathcal{S}$ is a \emph{simple system}.

The \emph{energy intensive variable} $p_i$ conjugate to $X_i$ is a partial derivative of $E$:
\begin{equation}
   p_0   :=   \left( \frac{  \partial E  }{  \partial S  } \right)_{X_m }  =  T
   \quad {\rm and} \quad
   p_i  :=  \left( \frac{  \partial E  }{  \partial X_i  } \right)_{S, X_{m \neq i} },
\end{equation}
wherein $X_m$ remains constant for all $m = 1, 2, \ldots, k  \neq  i$.
For example, the pressure $p$ is (the negative of) the energy intensive variable conjugate to volume:
\begin{equation}  
\label{eq:Pressure}
   p  =   - \left( \frac{  \partial E  }{  \partial V  } \right)_{S, X_m \neq V }.
\end{equation}
Equation~\eqref{eq:Pressure} follows, for a simple single-species, single-phase system, from
$E  =  TS  -  pV  +  \mu N$.
An analogue of Eq.~\eqref{eq:Pressure} has been used
to define pressure for quantum thermodynamic systems~\cite{Quan_09_Quantum}.
The \emph{entropy intensive variables} $F_i$ are partial derivatives of $S$:
\begin{equation}  \label{eq:SIntensive}
   F_i   :=   \left( \frac{  \partial S  }{  \partial X_i  } \right)_{ X_{m \neq i} }.
\end{equation}
Entropy intensive variables depend on energy intensive variables:
\begin{equation}  \label{eq:SIntensive2}
   F_0   =   \frac{1}{T},
   \quad {\rm and} \quad
   F_i   =   - \frac{ p_i }{ T }
   \; \; \forall  i = 1, 2, \ldots, k.
\end{equation}

Much thermodynamics can be cast equivalently in the energy and entropy representations. I will cast most thermodynamic resource theories in the energy representation, because Helmholtz theories have implicitly been cast so. The resource theory for entropy (or for information or nonuniformity) has implicitly been cast in the entropy representation.

%
%
\subsection{Thermodynamic potentials, Legendre transforms, and natural variables}
\label{section:FreeEnergyIntro}


\emph{Thermodynamic potentials}, or \emph{free energies}, resemble a spring's potential energy:
The ability to perform work can be stored in, and drawn from, thermodynamic potentials. One thermodynamic potential is $E$; the others result from Legendre-transforming $E$. Which variables are transformed depends on which properties of $\mathcal{S}$ an experimenter can control, or which properties remain constant. I will review Legendre transforms and four examples.

$E$ is a function of the independent variables $S$ and $X_{i = 1, 2, \ldots, k}$. Suppose that the experimenter can control $X_{j+1}, \ldots, X_k$ and the intensive variables 
$(T), p_1, p_2, \ldots, p_j$. The $(T)$ should be interpreted as follows: If an experimenter can hold $T$ constant during a particular transformation, $T$ should be included in the list that describes that transformation. If the experimenter cannot, $T$ should be excluded. This notation enables us to treat $T$ similarly to other independent variables while distinguishing $T$ as will become necessary.

A Legendre transform substitutes the controlled intensive variables for the uncontrolled extensive variables in the set of independent variables:
\begin{align}
   E[(T), p_1, p_2, \ldots, p_j]
   & :=  \label{eq:ArbFreeE0} 
            \inf_{(S), X_1, \ldots, X_j} 
            \left\{
            E  -   \left(    \left( \frac{ \partial E }{ \partial S } \right)_{X_{i}}   S  \right)
                -  \sum_{i = 1}^j   \left( \frac{ \partial E }{ \partial X_i } \right)_{(S), X_{m \neq i}}   X_i
            \right\}   \\
   & =   \label{eq:ArbFreeE}
             \inf_{(S), X_1, \ldots, X_j} 
            \left\{   E   -  ( TS )   -   \sum_{i=1}^j   p_i  X_i   \right\}.
\end{align}
The symbols enclosed in extra parentheses---the $T$ in Eq.~\eqref{eq:ArbFreeE0}, the second term on the RHS of Eq.~\eqref{eq:ArbFreeE0}, the subscript $S$, and the $TS$---participate in a given Legendre transformation if and only if $T$ is controllable. 
As in Callen's notation, $E[(T), p_1, \ldots, p_j ]$ is a function of the independent variables
$(T), p_1, \ldots p_j, X_{j+1}, \ldots, X_k$, called \emph{natural variables}~\cite{Callen85}. 
The most work that $\mathcal{S}$ can perform, or the most heat that $\mathcal{S}$ can release, on average while these variables remain constant equals the change in $E[(T), p_1, p_2, \ldots, p_j]$~\cite{Callen85,Reichl80}:
\begin{align}   \label{eq:ThermoWork}
   \langle W \rangle  =  \Delta   E[(T), p_1, p_2, \ldots, p_j].
\end{align}
$\mathcal{S}$ reaches equilibrium when $E[(T), p_1, p_2, \ldots, p_j]$ minimizes. 
Legendre-transforming all the extensive variables yields the trivial fundamental relation $E = 0$~\cite{Callen85}.
To simplify notation, I will usually omit the infimum from Legendre transforms. Extremization will be implied.

\begin{example}
Suppose the experimenter can control $T$ but not $S$. For example, suppose that $\mathcal{S}$ exchanges a heat with a bath. The \emph{Helmholtz free energy}
\begin{align}
   E[T]   :=   E  -  TS   =:   F
\end{align}
describes the system more naturally than $E$ does.
If $\mathcal{S}$ is simple and consists of one species and one phase, $F = F(T, V, N)$. The average, over many trials, of the work extractable from $\mathcal{S}$ during any constant-$(T, V, N)$ process satisfies $W \leq \Delta F$. $\mathcal{S}$ reaches equilibrium when the probability distribution over the possible microstates $i$ becomes the canonical ensemble,
$\{ P_i  =  e^{ - \beta E_i } / Z \}$.
The partition function $Z  :=  \sum_i  e^{ - \beta E_i }$ normalizes the distribution,
and $\beta  =  1 / ( \kB T )$.

Replacing not only $S$ with $T$, but also $N$ with $\mu$ yields the \emph{grand potential}. 
If $\mathcal{S}$ consists of one species and one phase,
\begin{equation}
   E[ T,  \mu ]   :=   E  -  TS  -  \mu N
   =: \Phi (T, \mu, V).
\end{equation}
$\Phi$ governs exchanges of heat and particles. If $T$, $\mu$, and $V$ remain constant, $\mathcal{S}$ reaches equilibrium when microstates' probabilities reach the grand canonical ensemble,
$\{ P_i  =  e^{ - \beta ( E_i  -  \mu n_i ) } / Z \}$. 
Here, the partition function $Z  :=  \sum_i  e^{ - \beta ( E_i  -  \mu n_i ) }$.

Many chemicals react at room temperature and atmospheric pressure. Constant-$T$, constant-$p$ processes are described by the \emph{Gibbs free energy,}
\begin{equation}
   E[ T, p ]   :=   E  -  TS  +  p V
   =:   G (T, p, N).
\end{equation}
(In lists of natural variables, I sometimes replace the energy intensive variable $-p$ with the pressure $p$.) $\mathcal{S}$ reaches equilibrium when microstates' probabilities reach the isothermal-isobaric ensemble $P_i  =  e^{ - \beta ( E_i  +  p v_i ) } / Z$.
A variation on $G$ suits a system characterized by a magnetic moment $\mathbf{m}$ and subject to an external magnetic field $\mathbf{B}$. If $T$ and $\mathbf{B}$ remain constant, the average extractable work does not exceed the change in the \emph{magnetic Gibbs potential,}
\begin{equation}
   E[ T,  \mathbf{B} ]   :=   E  -  TS  -  \mathbf{B} \cdot \mathbf{m}.
\end{equation}
\end{example}

%
%
\subsection{Massieu functions}
\label{section:Massieu}

Massieu functions are the entropy representations of thermodynamic potentials.
Massieu functions include, and result from Legendre-transforming, the entropy:
\begin{equation}
   S[(F_0), F_1, F_2, \ldots, F_j]   :=
   \sup_{(X_0), X_1, \ldots, X_j}   \left\{
   S    -  \left(   \left( \frac{ \partial S }{ \partial X_0 }  \right)_{ X_{i \neq 0} }  X_0   \right)
   -  \sum_{i = 1}^j   \left(  \frac{ \partial S }{ \partial X_i }  \right)_{ X_{m \neq i} }   X_i
   \right\}.
\end{equation}
If $(F_0), F_1, F_2, \ldots, F_j, X_{j+1}, \ldots, X_k$ remain constant, $\mathcal{S}$ attains equilibrium when
$S[(F_0), F_1, F_2, \ldots, F_j]$ maximizes. Physical significances of Massieu functions are less straightforward than the free energies' interpretations as work and heat. 

Massieu functions deserve mention because $S$ plays the role, in the resource theory for entropy (or information, or nonuniformity), played by free energies in other thermodynamic resource theories. In the entropy theory, resourcefulness is quantified by one-shot functions that converge, in the thermodynamic limit, to $S$. $S$ can be viewed as the entropy-representation ``thermodynamic potential'' suited to closed isolated systems~\cite{Belof09}. The energy representation of $S$ is $E$.

%
%
%
%
\section{Defining thermodynamic resource theories}
\label{section:DefineTheories}

Different families of thermodynamic resource theories correspond to exchanges of different quantities, to different sets of natural variables, to different free energies, and to different types of baths or external fields. The rest of this paper concerns systems whose natural variables include $T$, as well as closed isolated systems (whose natural variables do not). Each parenthesized $T$ should be interpreted as follows: If a closed isolated system is being modeled, the natural variables are extensive, so the $T$ should be regarded as absent. Otherwise, the $T$ should be regarded as present. 

I will explain how to specify a thermodynamic resource theory, then will illustrate with known resource theories. I will show how to model systems and states. Then, I will introduce a \emph{fixed-eigensubspace condition} that represents conservation of energy, particle number, etc.

\subsection{Families of thermodynamic resource theories}

%
%
\begin{table*}[t]
\begin{center}
\begin{tabular}{|c|c|c|c|c|c|c|} 
   \hline
        Thermo-
   &   Quantities   
   &   Natural   
   &   Parameters
   &   States
   &   Conserved   
   &   Free density   \\
         dynamic   
   &   exchanged   
   &   variables  
   &   that specify 
   &   
   &   operators   
   &   operator  \\
         potential
   &
   &
   &  a theory
   & 
   &
   & \\
   \hline \hline 
        Entropy ($S$) /
   &   None
   &   $E, V, N$ /
   &   None
   &   $\rho$
   &   None
   &   Microcanoncial  \\
         energy ($E$)   
   &   
   &    $S, V, N$  
   &    
   &   
   &   
   &   $(  \underbrace{   \frac{1}{d}, \ldots, \frac{1}{d}   }_d ) $  \\
   \hline 
        Helmholtz
   &   Heat
   &   $T, V, N$  
   &   $\beta$ 
   &   $(\rho, H)$
   &   $H_{\rm tot}$   
   &   Canonical   \\
        free energy ($F$)   
   &   
   &   
   &   
   &   
   &      
   &   $e^{-\beta {H}} / Z$   \\
   \hline 
        Grand po-   
   &   Heat,
   &   $T, \mu, V$
   &   $\beta, \mu$
   &   $(\rho, H, N)$
   &   $H_{\rm tot}$,
   &   Grand canonical   \\
         tential ($\Phi$)   
   &   particles
   &   
   &   
   &   
   &   $N_{\rm tot}$
   &   $e^{-\beta ( {H} - \mu N )} / Z$    \\
   \hline 
   \hline 
\end{tabular}
\end{center} 
\caption{\caphead{Commonly used thermodynamic potentials, and properties of the corresponding resource theories:} Sections~\ref{section:DefineTheories}-\ref{section:FreeStates} explain Columns 4-7. Each $Z$ denotes the partition function that normalizes the relevant state.}
\label{table:Classes}
\end{table*}

We wish to model interactions between a system $\mathcal{S}$ and a bath, using a thermodynamic resource theory. Suppose that the variables 
$(T), p_1, p_2, \ldots, p_j, X_{j+1}, X_{j + 2}, \ldots, X_k$ remain constant. The intensive variables $(T), p_1, \ldots, p_j$ characterize the bath accessible to the resource-theory agent. Each list of values of $(T), p_1, \ldots, p_j$ specifies one thermodynamic resource theory $\mathcal{T}^{ (\beta), p_1, \ldots, p_j }$ for the thermodynamic potential 
$E[(T), p_1, \ldots, p_j]$. 
All the resource theories characterized by the same intensive variables are mathematically equivalent and form a family. The bath exchanges with $\mathcal{S}$ the physical quantities represented by $(X_0), X_1, \ldots, X_j$. 

Attempting to specify the values of all the bath's intensive variables leads to a trivial resource theory. Suppose that, in a thermodynamic setting, each intensive variable $p_i$ changes by an amount $d p_i$. The \emph{Gibbs-Duhem Relation} interrelates the changes:
\begin{equation}
   \sum_{i=0}^k   X_i  \:  d p_i = 0
\end{equation}
\cite{Callen85}.
If you specify enough changes $d p_i$, you implicitly specify the rest of the $d p_j$. The Gibbs-Duhem Relation leads to the \emph{Gibbs Phase Rule}, which governs the number $f$ of independent intensive variables~\cite{Callen85,Alberty01}.  
For example, the intensive variables of a simple single-species, single-phase system are $(p_1, p_2, p_3)  =  (T, p, \mu)$. Two of these $p_i$'s can be specified independently. Each Helmholtz theory is defined by one $p_i$ ($\beta$); and each grand-potential theory, by two $p_i$'s ($\beta$ and $\mu$).
Attempting to specify more than $f$ intensive variables independently yields the trivial Fundamental Relation $E = 0$~\cite{Callen85}.

Table~\ref{table:Classes} summarizes the following examples.
\begin{example}
The most prevalent family of thermodynamic resource theories consists of \emph{Helmholtz theories} (e.g.,~\cite{Janzing00,BrandaoHORS13,FundLimits2,SkrzypczykSP13Extract}). Each Helmholtz theory models exchanges of energy with a heat bath characterized by an inverse temperature $\beta$. Free operations preserve energy, and canonical ensembles $e^{- \beta H} / Z$ specify free states. Results include the work $\Wext(R)$ extractable from one copy of a quasiclassical state $R$,\footnote{
Each quasiclassical state $R$ in any Helmholtz theory
has a density operator $\rho$ that commutes with
the state's Hamiltonian, $H$:
$[ \rho, H ] = 0$.}
and the work $\Wform(R)$ required to create one copy of $R$, via faulty protocols~\cite{FundLimits2}.
$\Wext(R)$ and $\Wform(R)$ depend on one-shot analogs of the Helmholtz free energy $F$.
\end{example}

\begin{example}
\emph{Grand-potential theories} model exchanges of energy and particles with heat-and-particle reservoirs. These theories were introduced in~\cite{YungerHalpernR14} to illustrate the generalization, detailed in this paper, of thermodynamic resource theories beyond Helmholtz theories. To specify a grand-potential theory, one specifies values of the bath's $\beta$ and chemical potential $\mu$. If particles of $m$ species in $\zeta$ phases are exchanged, one specifies $\mu^\alpha_i$ for all $i = 1, 2, \ldots, m$ and all $\alpha = 1, 2, \ldots, \zeta$.
Free operations conserve energy and particle number, and grand canonical ensembles are free. $\Wext(R)$ depends on a one-shot analog of the grand potential $\Phi$, and another one-shot analog bounds $\Wform(R)$.
\end{example}

\begin{example}
The \emph{entropy theory} models closed isolated systems. This theory was introduced in~\cite{HHOShort,HHOLong}, which covers the limit in which $n \to \infty$ copies of a state $\rho$ are processed. Single copies were analyzed in~\cite{BrandaoHNOW14,GourMNSYH13}. In~\cite{GourMNSYH13}, the entropy theory is called \emph{the resource theory of nonuniformity} or \emph{of informational nonequilibrium}. Free states are microcanonical ensembles $( \underbrace{ \frac{1}{d},   \frac{1}{d},   \ldots,   \frac{1}{d} }_d )$~\cite{HHOLong}.

The nonuniformity theory is equivalent to a Helmholtz theory in which all Hamiltonians are trivial: $H = 0$~\cite{HHOLong,GourMNSYH13}. Yet we shall see that the entropy theory models systems whose natural entropy-representation variables are $(E, V, N)$ (if each system is simple and consists of one species and one phase). $(E,V,N)$ are the natural variables of closed isolated systems. Because the variables natural to the entropy theory are not natural to Helmholtz theories, the entropy theory can be viewed as distinct from any Helmholtz theory.

The entropy theory has been thought to be a Helmholtz theory associated with a zero-temperature or infinite-temperature bath. But equating $\beta = \infty, 0$ fails to reduce Helmholtz theories' free operations to the entropy theory's free operations.
The free unitaries $U$ in Helmholtz theories
conserve the total Hamiltonian $H_{\rm tot}$:
$[U, H_{\rm tot}]  =  0$.
This requirement does not constrain
the free unitaries $U$ in the entropy theory.
See Sec.~\ref{section:FreeOps}.

The family of entropy theories differs from the other families in two ways. First, the entropy family (that models simple single-species, single-phase systems) contains only one theory. The family has no intensive natural variables whose different possible values would characterize different theories. Such an intensive variable would characterize a bath. Closed isolated systems do not interact with baths; so of course the entropy family lacks intensive natural variables.\footnote{
More precisely, the family of entropy theories associated with each set 
$\{X_0, X_1, \ldots, X_k \}$ of natural variables contains only one theory. One entropy theory models simple single-species, single-phase systems, whose natural variables are $(E, V, N)$; one entropy theory models simple two-species, one-phase systems, whose natural variables are $(E, V, N_1, N_2)$; etc.}
Second, the entropy theory is cast in the entropy representation. Central results---the amount of ``information'' extractable from, or needed to create, one copy of a state $\rho$---depend on entropies, rather than being work.
\end{example}

%
%
\subsection{Systems}
\label{section:States}

The extensive variables $X_0, X_1, \ldots, X_k$, and $S$ characterize the thermodynamic system $\mathcal{S}$. Let us associate an operator with each $X_i$, as Jaynes proposed~\cite{JaynesII57}. 
A Hamiltonian $H$ corresponds to $X_0  =  E$, for example, and the number operator corresponds to $N$. I will represent each operator (apart from $H$) with the symbol $X_i$ that represents the corresponding external variable. The association of the volume $V$ with an operator, which might sound unexpected, is discussed in Appendix~\ref{section:Volume}. 

\begin{definition}[System]
Each \emph{system} $\mathcal{S}$ in the thermodynamic resource theory $\mathcal{T}^{(\beta), p_1, \ldots, p_j}$ is specified by a Hilbert space $\mathcal{H}$ and by the \emph{system operators}
$\bm{(}   (H), X_1, \ldots, X_k   \bm{)}$. The system operators are Hermitian operators defined on $\mathcal{H}$.
\end{definition}
\noindent For simplicity, I will often assume that the $X_i$ commute with each other and have discrete spectra.

%
%
\subsection{States}

Recall that (in non-entropy thermodynamic resource theories)
$(T, p_1, \ldots, p_j, X_{j+1}, \ldots, X_k)$
are the natural variables of $\mathcal{S}$. The extensive-variable operators $X_1, \ldots, X_j$ conjugate to $p_1, \ldots, p_j$ will be called \emph{state operators}. To specify a state in a thermodynamic resource theory, one specifies state operators and a density operator.

\begin{definition}[State]
Let $\mathcal{S}$ denote a system, associated with the Hilbert space $\mathcal{H}$, in 
a thermodynamic resource theory $\mathcal{T}^{(\beta), p_1, \ldots, p_j}$.
Each possible \emph{state} of $\mathcal{S}$ is specified by a $(j+1)$-tuple or $(j+2)$-tuple
\begin{equation}
   R   :=   \bm{(}   \rho,  (H),  X_1,  \ldots,  X_j   \bm{)},
\end{equation}
wherein the density operator $\rho$ is a positive-semidefinite linear operator defined on $\mathcal{H}$.
\end{definition}

For simplicity, and in accordance with earlier works on thermodynamic resource theories~\cite{Janzing00,HHOLong,FundLimits2,YungerHalpernR14}, I will often focus on density operators and state operators that commute and that have discrete spectra. Commuting operators, and the states they define, are called \emph{quasiclassical}. Free operations can diagonalize a quasiclassical $\rho$, whose resourcefulness is encapsulated in a vector 
$r  =  (r_1, r_2, \ldots, r_d)$ of the eigenvalues of $\rho$~\cite{Janzing00,FundLimits2,YungerHalpernR14}. A quasiclassical state, therefore, will often be represented by $R  =  \bm{(}   r, (H), X_1, \ldots, X_j \bm{)}$.

\begin{example}
Helmholtz theories that model simple single-species, single-phase systems have the natural variables $(p_0, X_1, X_2) = (T, V, N)$.  The only intensive natural variable is $p_0$, which is conjugate to $E$, which corresponds to the operator $H$. Hence each state is defined by $(\rho, H)$, as in~\cite{Janzing00,FundLimits2,BrandaoHNOW14,YungerHalpernR14}. A grand-potential theory that models simple single-species, single-phase systems has the natural variables $(p_0, p_1, X_2)  =  (T, \mu, V)$. The $T$ and $\mu$ are conjugate to $E$ and $N$, so a state is specified by $(\rho, H, N)$, as in~\cite{YungerHalpernR14}. 
The entropy theory's natural variables are extensive: $(X_1, X_2, X_3)  =  (E, V, N)$.
To specify a state, one specifies only $\rho$, as in~\cite{HHOShort,HHOLong,GourMNSYH13}.
\end{example}

Permit me to introduce two more notations for states. First, recall that agents in resource theories can create \emph{free states} at zero cost. With each quasiclassical state $R  :=  \bm{(}   \rho, (H), X_1, \ldots, X_j   \bm{)}$ is associated a free state 
$G_R  :=  \bm{(}   \gamma_R, (H), X_1, \ldots, X_j   \bm{)}$, wherein $\gamma_R$ has the form detailed in Sec.~\ref{section:FreeStates}. In quasiclassical notation,
$G_R  =  \bm{(}   g_R, (H), X_1, \ldots, X_j   \bm{)}$.
Second, a composition of states 
$R :=  \bm{(}   \rho,  (X_{0_R}),  \ldots,  X_{j_R}   \bm{)}$ and 
$S  :=  \bm{(}   \sigma, (X_{0_S}),  \ldots,  X_{j_S}   \bm{)}$ will be denoted by
\begin{equation*}
   R  +  S   =   \bm{(}   \rho  \otimes  \sigma,   (X_{0_R}  +  X_{0_S}),   
                          X_{1_R}  +  X_{1_S},   \ldots,  X_{j_R}  +  X_{j_S}   \bm{)},
\end{equation*}
wherein
$X_{i_R}  +  X_{i_S}   =   ( X_{i_R} \otimes \id )   +   ( \id \otimes X_{i_S} )$.

%
%
%
%
\subsection{Fixed-eigensubspace condition}
\label{section:Fixed_eigen}

The system operators $X_{j+1}, \ldots, X_k$ do not specify states. In Helmholtz theories, for example, $V$ and $N$ do not characterize any state $(\rho, H)$. To understand the role played by nonstate system operators in thermodynamic resource theories, we can return to  thermodynamics. The change in a free energy bounds the average, over many trials, of the work extractable during a constant-$(p_1, \ldots, p_j, X_{j+1}, \ldots, X_k)$ process. For example, 
$W \leq \Delta F$. This constancy of the $X_{j+1}, \ldots, X_k$ suggests that the ``action'' in each thermodynamic resource theory takes place in one eigensubspace shared by $X_{j+1}, \ldots, X_k$.

\begin{proposition}[Fixed-eigensubspace condition]
\label{proposition:Fixed}
Let $R  :=  \bm{(}   \rho, (H), X_1, \ldots, X_j   \bm{)}$ denote the state of any system $\mathcal{S}$ specified by $\mathcal{H}$ and by 
$\bm{(}   (H), X_1, \ldots, X_k   \bm{)}$  in $\mathcal{T}^{(\beta), p_1, \ldots, p_j}$. The support $\supp(\rho)$ occupies an eigenspace $\mathcal{H}_0$ of $X_{j+1}$ that coincides with an eigenspace of $X_{j+2}$ and with an eigenspace of $X_{j+3}$ and so on for all 
$X_{i = j+1, \ldots, k}$:
\begin{equation} 
   \supp(\rho)  \subseteq  \mathcal{H}_0.
\end{equation}
All free unitaries $U$ (defined in Sec.~\ref{section:FreeOps}) preserve $\mathcal{H}_0$:
$\ket{\psi} \in \mathcal{H}_0
\; \Rightarrow \;
U \ket{\psi} \in \mathcal{H}_0$.
\end{proposition}

The fixed-eigensubspace condition has previously been mentioned in a specialized context. The resource theory for entropy has been portrayed as a Helmholtz theory in which all Hamiltonians are trivial: $H = 0$~\cite{HHOLong,GourMNSYH13,BrandaoHNOW14}. This portrayal was generalized~\cite{YungerHalpernR14}: The entropy theory is equivalent to a Helmholtz theory in which (i) every $\supp( \rho )$ occupies one energy eigensubspace $\mathcal{H}_0$ and (ii) $\mathcal{H}_0$ is the only subspace transformed nontrivially by free unitaries. The fixed-eigensubspace condition generalizes and sharpens the claims in~\cite{HHOLong,GourMNSYH13,BrandaoHNOW14,YungerHalpernR14}: The entropy theory (having the natural variables $E$, $V$, and $N$) is equivalent to a Helmholtz theory in which (i) every state's $\supp(\rho)$ occupies one eigensubspace $\mathcal{H}_0$ shared by the state's $H$, $V$, and $N$ and (ii) $\mathcal{H}_0$ is the only subspace transformed nontrivially by any free unitary. The entropy theory models physical transformations that conserve the total energy, volume, and particle number.

Not only the entropy theory, but also every other thermodynamic resource theory, can be viewed as having ``behind-the-scenes'' $X_i$'s. A Helmholtz-theory system corresponds to a $V$ and an $N$ that share an eigensubspace $\mathcal{H}_0$ in which $\supp(\rho)$ remains. ``All the action takes place'' in $\mathcal{H}_0$.  

The fixed-eigensubspace $X_i$ do not directly affect previously established results. Yet these $X_i$ matter for four reasons: (1) These $X_i$
clarify the relationship between the much-used Helmholtz and entropy theories. 
(2) The $N$'s in Helmholtz theories invite a reconsideration of what ``one-shot statistical mechanics'' means (as discussed in Sec.~\ref{section:Discussion}.\ref{section:OneShot}). 
(3) The existence of these $X_i$ can be viewed as an assumption implicit in thermodynamic resource theories. Identifying one's assumptions is advisable.
(4) These $X_i$ must exist for resource theories to model thermodynamics in all its natural-variable--containing glory.

%
%
%
%
\section{Free (equilibrating) operations}
\label{section:FreeOps}

In each resource theory, an agent can perform certain operations for free, without expending resources. The free operations in general thermodynamic resource theories are here termed \emph{equilibrating operations}, as in~\cite{YungerHalpernR14}.

\begin{definition}[Equilibrating operations]   
\label{definition:EquilOps}
Each \emph{equilibrating operation} $\mathcal{E}$ on any state 
$R  :=  \bm{(}   \rho,  (H_R),   X_{1_R},  \ldots,  X_{j_R}   \bm{)}$ consists of three steps:
\begin{enumerate}
   \item composition with any free state 
           $G  :=  \bm{(}   \gamma, (H_G), X_{1_G}  \ldots,  X_{k_G}   \bm{)}$;
   \item the transformation of $\rho \otimes \gamma$ by any unitary that commutes with $X_{i_R}  +  X_{i_G}$ for all $i = (0), 1, \ldots, j$ and that satisfies the fixed-eigensubspace condition; and
   \item the discarding of any component system $A$ associated just with its own system operators.
\end{enumerate}
$\mathcal{E}$ has the form
\begin{equation}
   R  \xmapsto{equil.}  \mathcal{E}(R)   
   =   \Big(    \tr_A ( U [ \rho  \otimes  \gamma ] U^\dag ),
                      \bm{(}  \tr_A ( H_R  +  H_G )   \bm{)},   
                          \tr_A ( X_{1_R}  +  X_{1_G} ),   \ldots,   
                          \tr_A ( X_{j_R}  +  X_{j_G} )   \Big),
\end{equation}
wherein
\begin{equation}   \label{eq:ConstrainU}
   [U,   X_{i_R}  +  X_{i_G}]  =  0   \; \;  \forall  i  =  (0), 1, \ldots, j
\end{equation}
and $U$ obeys the fixed-eigensubspace condition. That is, if $\supp (\rho \otimes \gamma)$ occupies an eigensubspace $\mathcal{H}_0$ shared by all the $X_{i_R} + X_{i_G}$, $U$ preserves $\mathcal{H}_0$:
\begin{equation}   \label{eq:FixedEigen}
   \ket{ \psi }  \in  \mathcal{H}_0
   \quad \Rightarrow \quad
   U \ket{ \psi } \in \mathcal{H}_0.
\end{equation}
\end{definition}

Equation~\eqref{eq:ConstrainU} distinguishes general thermodynamic resource theories from Helmholtz theories. $\mathcal{T}^{\beta, p_1, \ldots, p_j}$ might, \emph{prima facie}, appear equivalent to a Helmholtz theory whose Hamiltonians are replaced by the effective Hamiltonians  
\mbox{$\tilde{H}  =  H  -  \sum_{i = 1}^j  p_j X_j$.} But free unitaries in Helmholtz theories preserve $\tilde{H}$, whereas free unitaries in $\mathcal{T}^{\beta, p_1, \ldots, p_j}$ preserve $X_i$ for all $i = 0, 1, \ldots, j$. Because $[U, \tilde{H}] = 0$ does not imply $[U, X_i]  =  0$, more unitaries are free in the  Helmholtz theory than in $\mathcal{T}^{\beta, p_1, \ldots, p_j}$.

Free operations tend to evolve states toward free states. The free states, as shown in Sec.~\ref{section:FreeStates}, are equilibrium states. Hence the name \emph{equilibrating operations}.
Free operations induce a quasiorder on states, as explained in Sec.~\ref{section:Quasiorder}.

\begin{example}
Definition~\ref{definition:EquilOps} can be shown to reduce to the free operations defined previously for the Helmholtz and entropy theories (except that the fixed-eigensubspace condition does not appear explicitly in earlier definitions). In Helmholtz theories, just one extensive-variable operator, $H$, characterizes each state. By Definition~\ref{definition:EquilOps}, therefore, all free unitaries satisfy $[U, H_R  +  H_G]  =  0$. This restriction appears in earlier definitions of Helmholtz theories' free operations, which have been called \emph{thermal operations}~\cite{Janzing00,BrandaoHORS13}. 
In the entropy theory, no extensive-variable operators characterize states, so Definition~\ref{definition:EquilOps} does not restrict free unitaries. Neither do earlier definitions of entropy-theory free operations, which have been called \emph{noisy operations}~\cite{HHOShort,GourMNSYH13}.

Definition~\ref{definition:EquilOps} illustrates (the previously known reason) why the entropy theory is not a Helmholtz theory in which $\beta = 0$ or $\infty$. Equation~\ref{eq:ConstrainU} constrains the $U$'s that are free in Helmholtz theories but not the $U$'s that are free in the entropy theory. Setting $\beta = 0, \infty$ in a Helmholtz theory does not lift the constraint---does not reduce Helmholtz-theory equilibrating operations to entropy-theory equilibrating operations. 
\end{example}

%
%
%
%
\section{Free (equilibrium) states}
\label{section:FreeStates}

The definition of equilibrating operations refers to free states but not to the forms that free states can assume. Using Definition~\ref{definition:EquilOps}, I will show that the only density operators that can be free in nontrivial quasiclassical thermodynamic resource theories are equilibrium ensembles. A resource theory will be called \emph{trivial} if free operations alone can generate states that do not appear explicitly as free states in the definition of free operations. Recall from Eqs.~\eqref{eq:SIntensive} and~\eqref{eq:SIntensive2} that $F_i$ denotes the entropy intensive variable conjugate to $X_i$. 

\begin{theorem}   \label{theorem:FreeStates}
Consider a thermodynamic resource theory 
$\mathcal{T}^{(\beta), p_1, \ldots, p_j}$ in which all states are quasiclassical.
Every free state has the form
\begin{align}
   G :=
   \bm{(} g, (X_0), X_1, \ldots, X_j \bm{)},
\end{align}
wherein element $\alpha$ of $g$ has the form
\begin{align}   \label{eq:FreeState}
   g_\alpha =  e^{ - \frac{1}{ k_B } 
           \left(  F_0 x_{0_{ \alpha }}    +  F_1 x_{1_{ \alpha }}    
                   + \ldots + F_j x_{j_{ \alpha }} \right) } / Z ,
\end{align}
the partition function $Z$ normalizes $g$, and $x_{i_{ \alpha }}$ denotes the $\alpha^{\rm th}$ eigenvalue of operator $X_i$. ($x_{i_\alpha}$ corresponds to the eigenstate $\ket{\alpha}$ of $X_i$ that equals the eigenstate associated with eigenvalue $x_{j_\alpha}$ of $X_j$, for all $i$ and $j$.) If any other state were free, $\mathcal{T}^{(\beta), p_1, \ldots, p_j}$ would be trivial.
\end{theorem}

The proof of Theorem~\ref{theorem:FreeStates} generalizes from~\cite{YungerHalpernR14} almost trivially and appears in Appendix~\ref{section:FreeStateApp}. As noted in~\cite{YungerHalpernR14}, the proof offers an operational alternative to canonical-ensemble derivations that depend on postulates such as the Fundamental Assumption of Statistical Mechanics. According to the Fundamental Assumption, microcanonical ensembles represent isolated systems' equilibrium states. The microcanonical form of the free states in the entropy theory can be derived from the definition of free operations and from the theory's nontriviality~\cite{HHOLong}. This operational derivation replaces the Fundamental Assumption, which has drawn criticism (e.g.,~\cite{JaynesII57,Lenard78}).

\begin{example} \label{ex:FreeStates}
If $\mathcal{T}^{(\beta), p_1, \ldots, p_j} = \mathcal{T}$ denotes the entropy theory, the argument of the exponential in Eq.~\eqref{eq:FreeState} vanishes. The free state $g$ is the microcanonical ensemble $( \underbrace{ \frac{1}{d}, \ldots, \frac{1}{d} }_d )$, as in~\cite{HHOLong}. 
In Helmholtz theories, free states' density operators have the form $e^{- \beta H} / \tr(e^{- \beta H})$; and in grand-potential theories, $e^{- \beta (H - \mu N) }/ \tr(e^{- \beta (H - \mu N) })$. 
\end{example}


%
%
%
%
\section{Quasiorder of states}
\label{section:Quasiorder}

Equilibrating operations induce a quasiorder on states. The quasiorder on quasiclassical states is shown to be equivalent to $d$-majorization relative to equilibrium states, termed \emph{equimajorization}. Rescaled Lorenz curves illustrate equimajorization. Much of this section immediately generalizes~\cite{YungerHalpernR14}, which generalizes~\cite{Janzing00,FundLimits2}.

%
%
\subsection{Equivalence of two quasiorders}

A \emph{quasiorder} on a set $\mathscr{S}$ is a binary operation  $\leq$   that satisfies reflexivity and transitivity: For all $A, B, C \in \mathscr{S}$, 
$A  \leq  A$; and if $A \leq B$ and $B \leq C$, then $A \leq C$~\cite{MarshallOA10}. 
Equilibrating operations define a quasiorder $\xmapsto{equil.}$ on the states in each thermodynamic resource theory. If equilibrating operations can transform $R$ into $S$, then $R \xmapsto{equil.} S$.

The quasiorder on quasiclassical states will be shown to be equivalent to $d$-majorization. 
A matrix $M$ is called \emph{$d$-stochastic} if it preserves some vector $d$:
\begin{align}
   M  d   =   d,   
   \quad   {\rm and}   \quad
   \sum_{i}  M_{ij}  =  1   \; \;  \forall i
\end{align}
\cite{MarshallOA10}.
A vector $r$ \emph{$d$-majorizes} a vector $s$, $r \geq_d s$, if some $d$-stochastic matrix $M$ maps $r$ to $s$:   $M  r   =   s$~\cite{MarshallOA10}. 
The $d$'s relevant to $\mathcal{T}^{(\beta), p_1, \ldots, p_j}$, are the equilibrium states relative to $\bm{(}  (T), p_1, \ldots, p_j  \bm{)}$. 
Matrices $M$ that preserve such equilibrium states will be called \emph{equistochastic}, as in~\cite{YungerHalpernR14}. $d$-majorization relative to the uniform distribution $( \underbrace{ \frac{1}{d}, \ldots, \frac{1}{d} }_d )$ is majorization, the quasiorder in the resource theory for entropy~\cite{HHOShort,HHOLong,GourMNSYH13}.

%
%
\begin{definition} \label{definition:Equimajorization}
Let $\mathcal{T}^{(\beta), p_1, \ldots, p_j}$ denote a thermodynamic resource theory in which the states $R := \bm{(} r,  (H),  X_{1}, \ldots, X_{j}  \bm{)}$ and 
$S  :=  \bm{(}   s,   (H),  X_{1},   \ldots,   X_{j}   \bm{)}$  share their extensive-variable state operators. $R$ \emph{equimajorizes} $S$, written $R  \succ_{(\beta), p_1, \ldots, p_j}   S$, if some stochastic matrix $M$ that preserves the equilibrium state 
$G_R  :=  \bm{(}   g_R,  (H),  X_1,   \ldots,   X_j   \bm{)}$ maps $r$ to $s$. That is, if
\begin{equation}
   M r  =  s,
   \quad
   M g_R  =  g_R,
   \quad {\rm and} \quad
   \sum_i   M_{ij}   =   1,
\end{equation}
then $R  \succ_{(\beta), p_1, \ldots, p_j}   S$.

Let $R' := \bm{(} r',  (H_R),  X_{1_R}, \ldots, X_{j_R} \bm{)}$ and 
$S' := \bm{(} s',  (H_S),  X_{1_S}, \ldots, X_{j_S} \bm{)}$ denote states that do not share all their extensive-variable state operators. $R'$ equimajorizes $S$ if some stochastic matrix $M$ that preserves the equilibrium state
\begin{equation*}
   G_R  +  G_S  
   :=  \bm{(}   g_R  \otimes  g_S,    (H_R + H_S),   X_{1_R}   +   X_{1_S},    \ldots, 
         X_{j_R}  +  X_{j_S}   \bm{)}
\end{equation*}
maps $r'  \otimes  g_{S}$  to  $g_R  \otimes  s'$:
\begin{equation}
   M (r'  \otimes  g_{S})  =  (g_R  \otimes  s'),
   \quad
   M (g_R   \otimes  g_S)  =  (g_R   \otimes  g_S),
   \quad {\rm and} \quad
   \sum_i   M_{ij}   =   1
\end{equation}
implies  $R'  \succ_{(\beta), p_1, \ldots, p_j}  S'$.
\end{definition}

The second definition invokes a technique used in~\cite{HHOLong,GourMNSYH13} to compare entropy-theory states $r'$ and $s'$ defined on Hilbert spaces that have different dimensions. $r'$ can be composed with the free state $g_S$; the composite $r' \otimes g_S$, compared to $g_R \otimes s'$; and the first subsystem, discarded. The two quasiorders in $\mathcal{T}^{(\beta), p_1, \ldots, p_j}$ are equivalent.

%
%
\begin{theorem}   \label{theorem:Quasiorder}
Let $R$ and $S$ denote any quasiclassical states in $\mathcal{T}^{(\beta), p_1, \ldots, p_j}$. An equilibrating operation maps $R$ to $S$ if and only if $R$ equimajorizes $S$:
\begin{equation}
   R \xmapsto{equil.}  S
   \quad  \Longleftrightarrow  \quad
   R \succ_{(\beta), p_1, \ldots, p_j} S.
\end{equation}
\end{theorem}

\begin{proof}
The proof immediately generalizes the proof of~\cite[Theorem 2]{YungerHalpernR14}, which concerns grand-potential theories. The grand-potential proof generalizes the proof of~\cite[Theorem 5]{Janzing00}, which concerns Helmholtz theories. The Helmholtz-theory proof centers on (i) a free operation $\pi_n$ that conserves energy and (ii) the equality of the energies of pure states that share a characteristic denoted by the pair $(u, v)$ of vectors. (In~\cite{Janzing00}, $(u, v)$ is denoted by $(r, s)$.) In grand-potential theories, an analog of $\pi_n$ conserves energy and particle number, and pure states that have the same $(u, v)$ have the same energy and particle number~\cite{YungerHalpernR14}. By the same token, in $\mathcal{T}^{(\beta), p_1, \ldots, p_j}$, an analog of $\pi_n$ that conserves multiple extensive properties can be constructed. Pure states associated with the same $(u, v)$ have the same set of eigenvalues of $(H), X_1, \ldots, X_j$. Replacing two statements about energies, in the proof of~\cite[Theorem 5]{Janzing00}, with statements about multiple $X_i$'s yields a proof of Theorem~\ref{theorem:Quasiorder}.
\end{proof}

%
%
\subsection{Rescaled Lorenz curves}

Rescaled Lorenz curves illustrate the quasiorder on states. The curves have been introduced into Helmholtz theories~\cite{FundLimits2}, the entropy theory~\cite{GourMNSYH13}, and grand-potential theories~\cite{YungerHalpernR14}. These resource-theory applications have roots in works by Ruch, Schranner, Seligman, and others (e.g.,~\cite{RuchSS78}). Extant results are generalized concisely below.

%
%
\begin{figure}[hbt]
\centering
\includegraphics[width=.55\textwidth, clip=true]{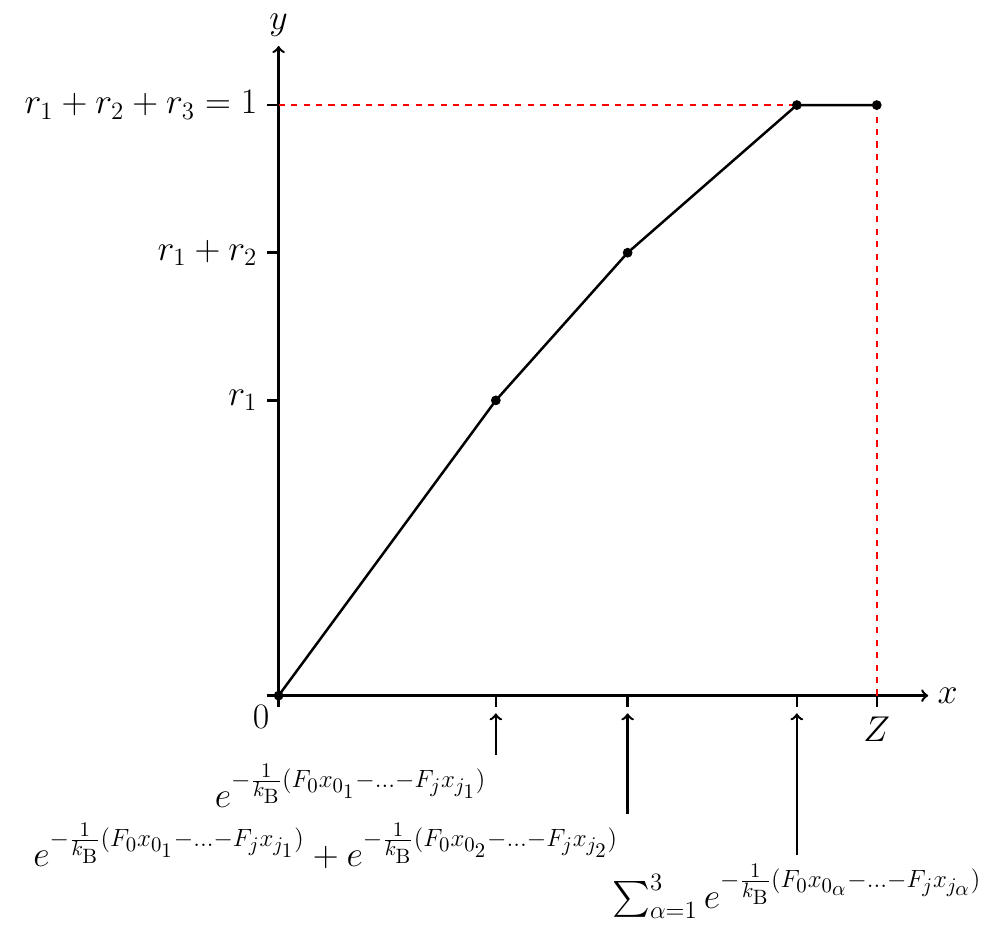}
\caption{
\caphead{Rescaled Lorenz curve $L_R$:} Visualization of a state 
$R := (r,  H,  X_1, \ldots, X_j)$ in a general thermodynamic resource theory
$\mathcal{T}^{(\beta), p_1, \ldots, p_j}$. $x_{i_{\alpha} }$ denotes eigenvalue $\alpha$ of the state operator $X_i$, and $F_i$ denotes the entropy intensive variable conjugate to $X_i$ [Eqs.~\eqref{eq:SIntensive} and~\eqref{eq:SIntensive2}]. If
$\mathcal{T}^{(\beta), p_1, \ldots, p_j}  =  \mathcal{T}$ denotes the entropy theory, each exponential's argument is replaced with a zero.}
\label{fig:StateCurve}
\end{figure}

\begin{definition}[Rescaled Lorenz curve]
Let $R   :=   \bm{(}   r,  (H),  X_{1},  \ldots,  X_{j}   \bm{)}$ denote any quasiclassical state in $\mathcal{T}^{(\beta), p_1, \ldots, p_j}$.
Let $G_R$ denote the corresponding equilibrium state, and let $Z$ denote the partition function of $G_R$. Suppose that $X_i$, for each $i = (0), 1, \ldots, j$, has $d$ discrete, not necessarily distinct, eigenvalues 
$x_{i_1}, \ldots, x_{i_d}$.
The $X_i$ eigenstate associated with eigenvalue $x_{i_\alpha}$ is assumed to equal the $X_j$ eigenstate associated with eigenvalue $x_{j_\alpha}$, for all $i$, $j$, and $\alpha$. Let $r  =  (r_1, \ldots, r_d)$ be ordered such that
\begin{align}
   r_1  e^{ \frac{1}{\kB}   [ 
      (F_0  x_{0_1})  +  F_1 x_{1_1}  +  \ldots  +  F_j x_{j_1}   ] }
   \geq   \ldots   \geq
   r_d  e^{ \frac{1}{\kB}  [ 
      (F_0  x_{0_d})  +  F_1 x_{1_d}  +  \ldots  +  F_j x_{j_d}  ] },
\end{align}
wherein $F_i$ denotes the entropy intensive variable conjugate to $X_i$.

Define the point $P_0  :=  (0, 0)$ and the points
\begin{align}   \label{eq:Elbows}
   P_m  :=   \left(
   \sum_{\alpha = 1}^m   e^{ - \frac{1}{\kB}   [ 
      (F_0  x_{0_\alpha})  +  F_1 x_{1_\alpha}  +  \ldots  +  F_j x_{j_\alpha}  ] },
   \sum_{\alpha = 1}^m   r_\alpha   \right)
\end{align}
for all $m = 1, \ldots, d$. The \emph{Lorenz curve $L_R^{(\beta), p_1, \ldots, p_j}$
for $R$, rescaled relative to $G_R$}, is the piecewise linear curve, defined on 
$x \in [0, Z]$, that consists of $P_{m = 0, 1, \ldots, d}$ and that interpolates between $P_m$ and 
$P_{m+1}$ for all $i = 0, \ldots, d-1$. 

If $\mathcal{T}^{(\beta), p_1, \ldots, p_j}  =  \mathcal{T}$ denotes the entropy theory, the argument of each exponential in Eq.~\eqref{eq:Elbows} is replaced with a zero.
\end{definition}

\noindent For simplicity, I will sometimes denote $L_R^{(\beta), p_1, \ldots, p_j}$ by $L_R$. 
An example curve appears in Fig.~\ref{fig:StateCurve}.

Rescaled Lorenz curves encapsulate states' resourcefulness, in the sense detailed in~\cite{FundLimits2,GourMNSYH13}. The more 
$L_R$ bends outward from the straight line that represents the equilibrium state $G_R$, the more value $R$ has. Many functions, called \emph{monotones}, quantify a state's value~\cite{BrandaoHNOW14,GourMNSYH13,YungerHalpernR14}. Monotones include the work needed to create, and the work extractable from, a state.

%
%
\begin{figure}[hbt]
\centering
\includegraphics[width=.45\textwidth, clip=true]{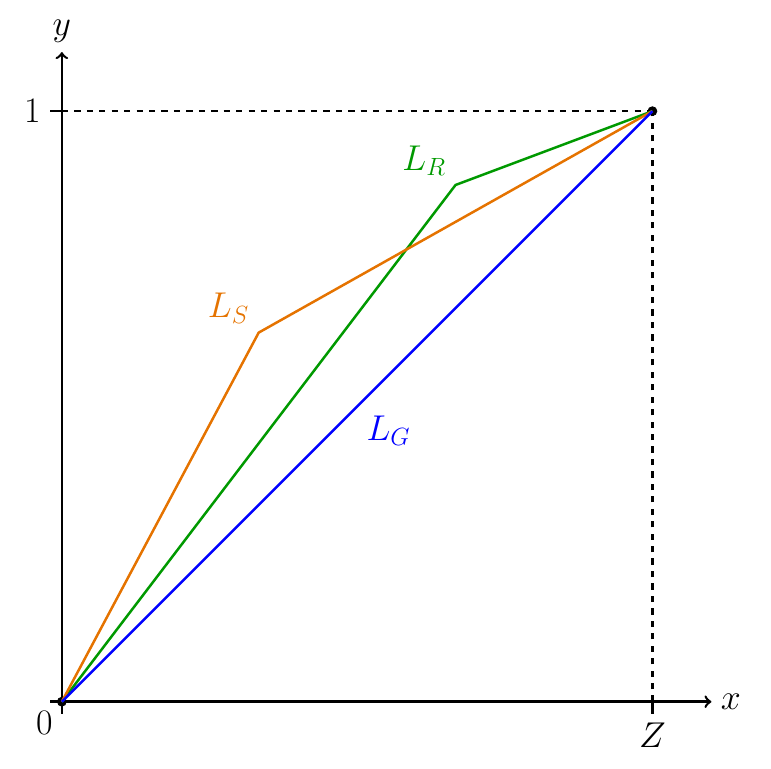}
\caption{\caphead{Comparison of rescaled Lorenz curves:}   The rescaled Lorenz curves $L_R$ and $L_S$ lie above the curve $L_G$. Equilibrating operations, therefore, can transform the states $R$ and $S$ into the free state $G$. Because $L_R$ lies partially below and partially above $L_S$, equilibrating operations can transform neither $R$ into $S$ nor $S$ into $R$. For simplicity, the state operators of $R$ (which represent extensive variables) are assumed to equal those of $S$ and those of $G$.}
\label{fig:CompareStates}
\end{figure}

%
%
\begin{theorem}   \label{theorem:LorenzCurves}
Let $R   :=   \bm{(}   r,  (H),  X_{1},  \ldots,  X_{j}   \bm{)}$ and 
$S  :=  \bm{(}   s,  (H),  X_{1},  \ldots,  X_{j}   \bm{)}$ denote quasiclassical states that share their extensive-variable state operators in the thermodynamic resource theory $\mathcal{T}^{(\beta), p_1, \ldots, p_j}$. Some equilibrating operation $\mathcal{E}$ can transform $R$ into $S$ if and only if the Lorenz curve for $R$, rescaled relative to the corresponding equilibrium state $G$, fails to dip below the rescaled Lorenz curve for $S$:
\begin{equation}
   \exists   \,   \mathcal{E}  \:  :  \:
   \mathcal{E}(R)  =  S
   \quad \Longleftrightarrow \quad
   L_R^{(\beta), p_1, \ldots, p_j}(x)
      \geq   L_S^{(\beta), p_1, \ldots, p_j}(x)
   \quad   \forall   x  \in  [0,  Z],
\end{equation}
wherein $Z$ denotes the partition function of $G$.

Let $R'  :=  \bm{(}   r,  (H_{R'}),  X_{1_{R'}},  \ldots,  X_{j_{R'}}   \bm{)}$ and 
$S'  :=  \bm{(}   s,  (H_{S'}),  X_{1_{S'}},  \ldots,  X_{j_{S'}}   \bm{)}$ denote quasiclassical states that do not share all their  $X_{i= (0), 1, \ldots, j}$. An equilibrating operation can transform $R'$ into $S'$ if and only if the Lorenz curve for  $R' + G_{S'}$, rescaled relative to $G_{R'} + G_{S'}$, fails to dip below the rescaled Lorenz curve for $G_{R'} + S'$:
\begin{align}
   \exists   \,   \mathcal{E}  \:  :  \:
   \mathcal{E}(R')  =  S'
   \quad \Longleftrightarrow \quad
   L_{R'  +  G_{S'} }^{(\beta), p_1, \ldots, p_j}(x)
      \geq   L_{G_{R'}  +  S'}^{(\beta), p_1, \ldots, p_j}(x)   
   \quad   \forall   x  \in  [0,  Z'],
\end{align}
wherein $Z'$ denotes the partition function for $G_{R'} + G_{S'}$.
\end{theorem}

\begin{proof}
The proof immediately generalizes the proof of~\cite[Proposition 3]{YungerHalpernR14}. That Proposition 3 contains the manifestation, in grand-potential theories, of Theorem~\ref{theorem:LorenzCurves}. The grand-potential proof does not depend on extensive-variable operators ($X_0 = H$ or $X_1 = N$) directly. The $X_i$'s affect the proof only insofar as $g_R$ depends on the $X_i$'s. Hence the proof of~\cite[Proposition 3]{YungerHalpernR14} can be restated in $\mathcal{T}^{(\beta), p_1, \ldots, p_j}$ under the assumption that $g_R$ has the form of Eq.~\eqref{eq:FreeState}.
\end{proof}
\noindent An illustration of Theorem~\ref{theorem:LorenzCurves} appears in Figure~\ref{fig:CompareStates}.

%
%
%
%
\section{Work and batteries}
\label{section:Battery}

To calculate the work transferred during a state conversion, we must define work. In conventional statistical mechanics, work is defined as an integral along a phase-space path. How to define a quantum analog has been debated~\cite{TalknerLH07}. In thermodynamic resource theories, work has been defined in four ways. I will fuse two of the definitions and recast them in terms of generalized thermodynamic resource theories. I will focus on energy-representation resource theories, whose natural variables include $\beta$.

Batteries have been modeled in four ways in thermodynamic resource theories. A two-level \emph{work bit} appears in~\cite{FundLimits2,BrandaoHNOW14}. Though mathematically simple, the bit can be difficult to use in practice, if the precise amount of work to be transferred during some state conversion is not known while the bit is prepared~\cite{SkrzypczykSP13Extract}.   
In~\cite{SkrzypczykSP13Extract}, a quasiclassical weight whose height changes stores gravitational potential energy. In~\cite{YungerHalpernR14}, a battery is any quasiclassical system whose energy levels are finely spaced.
Coherences have been addressed in~\cite{FrenzelDR14}. For simplicity, this paper focuses on quasiclassical batteries generalized from~\cite{YungerHalpernR14,SkrzypczykSP13Extract}. The generalization will be illustrated when the weight in~\cite{SkrzypczykSP13Extract} is modeled with gravitational thermodynamic resource theories, which will be shown to be mathematically equivalent to grand-potential theories.

%
%

Consider a quasiclassical battery $B$ in $\mathcal{T}^{\beta, p_1, \ldots, p_j}$. The battery stores resourcefulness insofar as the density operator differs from the corresponding equilibrium state's. For simplicity, I assume that $B$ occupies only energy eigenstates (or states close, in $L^1$ distance, to energy eigenstates. See Sec.~\ref{section:Error_prob}\footnote{
Part of the derivation of the $\varepsilon$-approximate work $\Wext (R)$
extractable from a state $R$ involves such a close state.
A free operation that leaves $B$ in such a state is constructed.
See Eq.~(C19) on p. 13 of~\cite{YungerHalpernR14}.}):
\begin{equation}
   B_E  :=  ( \ketbra{E}{E},  H,  X_1,  \ldots,  X_j).
\end{equation}
Little distinguishes the energy operator $H$ from the other $X_i$'s. Each $X_i$ represents some physical degree of freedom (DOF)~\cite{JaynesI57}. Resourcefulness can be stored in particle number, electric polarization, gravitational mass, etc. Angular momentum is treated as a resource in~\cite{VaccaroB11,BarnettV13}.

Suppose that $W$ denotes a positive number. Storing work in $B$ maps $B_E$ to $B_{E + W}$, and withdrawing work maps $B_{E+W}$ to $B_{E}$.
I will define states' work yields and work costs as in~\cite{YungerHalpernR14}: The \emph{maximum work $W_{\rm gain}( R )$ extractable from a state $R$} is the greatest value of $W$ for which 
$R + B_{E}  \xmapsto{equil.}   B_{E+W}$. The \emph{minimum work cost 
$W_{\rm cost}(R)$ of creating $R$} is the least value of $W$ for which 
$B_{E+W}  \xmapsto{equil.}  R  +  B_{E}$.

%
%

\begin{example}[Weights in gravitational resource theories]

Batteries can manifest as weights described by thermodynamic resource theories that model gravitational potential energy. Such theories, we shall see, are equivalent to grand-potential theories. 

Gravitational fields are modeled as follows in conventional thermodynamics~\cite{Alberty01,Guggenheim59,KirkwoodO61}. 
Consider a massive particle at a height $h$ in a uniform gravitational field sourced by external masses. The particle's mass $m$ serves as an extensive variable. The gravitational potential 
$\phi = gh$ serves as an intensive variable.

More generally, consider a system that contains $q$ chemical species. Let $m_i$ denote the mass of each species-$i$ particle. Each particle sits at some height. The set 
$\{h_0, h_1, \ldots, h_\zeta \}$ of possible heights will be approximated as discrete, similarly to the set in~\cite{SkrzypczykSP13Extract}. The ground is labeled as $h_0$, and species-$i$ particles at different heights are viewed as occupying different phases~\cite{Guggenheim59,KirkwoodO61,Alberty01}. The number of species-$i$ particles at height $h_\alpha$ is denoted by $N_i^\alpha$. The system has the gravitational potential energy
\begin{equation}   \label{eq:BatteryH}
   g  \sum_{\alpha = 1}^\zeta   \sum_{i = 1}^q    m_i   N_i^\alpha  h_\alpha.
\end{equation}

Gravitational potential energy can be cast as a contribution to chemical energy~\cite{Guggenheim59,KirkwoodO61,Alberty01}.  The chemical potential $\mu^\alpha_i$ that characterizes the phase $\alpha$ of species $i$ in the absence of gravitational and electric fields is called the \emph{standard chemical potential}. Combining $\mu_i^\alpha$ with a gravitational term yields the chemical potential $\tilde{\mu}^\alpha_i$:
\begin{align}
   \sum_{\alpha = 1}^\zeta   \sum_{i = 1}^q   \mu_i^\alpha N_i^\alpha
            +   g   \sum_{\alpha = 1}^\zeta    m_i   N_i^\alpha  h_\alpha  
        &  =   \sum_{\alpha = 1}^\zeta   \sum_{i = 1}^q
                  (\mu_i^\alpha   +   m_i  g  h_\alpha)   N_i^\alpha  \\
        &  =   \label{eq:GravE}
                 \sum_{\alpha = 1}^\zeta   \sum_{i = 1}^q
                  \tilde{\mu}_i^\alpha   N_i^\alpha.
\end{align}

In the resource-theory framework, grand-potential theories can model gravitational energy. Each possible state of $B$ has the form
\begin{equation}   \label{eq:BatteryState}
   B_E  :=  (\rho,  H,  N_0^1,  \ldots,  N_0^\zeta,   \ldots,   
      N_q^1,  \ldots,   N_q^\zeta ),
\end{equation}
wherein $\rho$ specifies the values of the $N_i^\alpha$ and $H$ represents the internal energy.

An analogous formalism describes the electrochemical energy of charged particles in an electric field~\cite{Guggenheim59,KirkwoodO61,Alberty01}. Grand-potential theories, featuring states specified by $\{ N_i^\alpha \}$ as well as by $H$, model fully the weight used often as a battery~\cite{SkrzypczykSP13Extract,Szilard29}. Using the theory suited to the battery's free energy, rather than a Helmholtz theory, becomes important if the battery has nongravitational DOFs. Such DOFs would be represented by extra extensive-variable operators in Eq.~\eqref{eq:BatteryState}. Free unitaries must separately conserve the operator associated with each DOF. This separate conservation appears in the general definition of equilibrating operations but not in the definition of Helmholtz theories' free operations.
\end{example}

%
%
%
%
\section{One-shot work yield and work cost}
\label{section:OneShotWork}

The most work extractable, on average over $n \to \infty$ trials, during a thermodynamic process in which
$(T), p_1, \ldots, p_j, X_{j+1}, \ldots, X_k$ remain constant is 
$W = \Delta E[ (T), p_1, \ldots, p_j ]$. When work is extracted from few copies of a state $R$, knowing the optimal average work might prove less useful than knowing the amount $\Wext(R)$ of work outputted by a realistically faulty implementation of an optimal protocol. 
$\Wext(R)$ can be calculated, and the work $\Wform(R)$ required to create one copy of an approximation to $R$ can be bounded, with the \emph{hypothesis-testing entropy} $\Dh^\varepsilon$. Below,
$\Dh^\varepsilon$ and the failure tolerance $\varepsilon$ are defined as in~\cite{YungerHalpernR14}. $\Wext(R)$ and $\Wform(R)$ are quantified by immediate generalizations from~\cite{YungerHalpernR14}. Most of this section (apart from ``Background'') concerns free energies and so not the entropy theory. States are assumed to be quasiclassical.

%
%
\subsection{Background}
\label{section:Error_prob}

Consider a protocol for creating one copy of a \emph{target state} 
$R  :=  \bm{(}   r,  (H),  X_1, \ldots, X_j   \bm{)}$. A realistic trial generates an \emph{actual output}
$\tilde{R}  :=  \bm{(} \tilde{r},  (H),  X_1, \ldots, X_j   \bm{)}$ that might differ from $R$. For simplicity, $R$ and $\tilde{R}$ are assumed to differ only due to their state vectors. The $L^1$ distance $\frac{1}{2}  ||  r  -  \tilde{r}  ||_1$ quantifies the discrepancy. If the $L^1$ distance falls below the \emph{tolerance} $\varepsilon  \in  [0, 1]$, the trial succeeds,
\begin{equation}
   \frac{1}{2}  ||  r  -  \tilde{r}  ||_1
   \leq  \varepsilon,
\end{equation}
and $\tilde{R}$ is said to be $\varepsilon$-\emph{close} to $R$:
$R  \approx_\varepsilon  \tilde{R}$. 
The existence of a free operation that maps $R$ to a state $\tilde{S} \approx_\varepsilon S$ is denoted by $R  \succ_{(\beta),  p_1,  \ldots,  p_j}^\varepsilon  S$.
The work $\Wext(R)$ that is \emph{$\varepsilon$-extractable} from $R$ is the greatest $W$ that satisfies $R  +  B_E   \succ_{(\beta),  p_1,  \ldots,  p_j}^\varepsilon   B_{E+W}$ for any 
$E > 0$. The \emph{$\varepsilon$-work cost} $\Wform(R)$ is the least $W$ that satisfies
$B_{E+W}  \succ_{(\beta),  p_1,  \ldots,  p_j}^\varepsilon   R  +  B_E$.

\emph{Hypothesis testing} has been formulated in quantum-information contexts as follows~\cite{Hiai_91_Proper}. Imagine being handed a quantum state and being told that the state is $\rho$ or $\gamma$. Knowing the forms of $\rho$ and $\gamma$, we wish to ascertain which state we were handed. We perform a positive operator-valued measurement (POVM) 
$\{ Q,  \id - Q \}$. If the measurement yields outcome $Q$, the state is probably $\rho$; if $\id - Q$, then $\gamma$. 

Errors of two types can occur. We commit a Type I error if we were given $\rho$ but $\id - Q$ obtains, such that we guess the state is $\gamma$. The Type I-error probability equals ${\rm Tr} \bm{(} (\id - Q) \rho \bm{)}$. A Type II error occurs if we were given $\gamma$ but $Q$ obtains, such that we guess the state is $\rho$. The Type II-error probability Tr$(Q \gamma )$ can be constrained by $\varepsilon$ and $Q$. Consider choosing $Q$ such that the Type I-error probability is at most $\varepsilon$:  
${\rm Tr} \bm{(} (\id - Q) \rho \bm{)}  \leq  \varepsilon$. The \emph{optimal} $Q$ minimizes the Type II-error probability. The hypothesis-testing entropy $\Dh^\varepsilon$ is defined in terms of the optimal Type II-error probability $b_\varepsilon(\rho   ||   \gamma)$.

\begin{definition} \label{eq:DEpsH}
Let $\rho$ and $\gamma$ denote density operators defined on $\mathcal{H}$.
Consider distinguishing between $\rho$ and $\gamma$ by hypothesis test. 
The optimal Type II-error probability associated with any hypothesis test that has a Type I-error probability of at most  $\varepsilon$ is
\begin{align} \label{eq:htprimal}
   b_\varepsilon(\rho   ||   \gamma)
   :=   \mathop{\min_{\tr (Q\rho) \geq 1 - \varepsilon}}_{0\leq Q\leq \id}\tr(Q\gamma).
\end{align}
The hypothesis-testing relative entropy is defined as
\begin{align}  \label{eq:Primal}
   \Dh^\varepsilon(\rho   ||   \gamma)
   :=   -\ln b_\varepsilon(\rho   ||   \gamma)
\end{align}
or, equivalently, by $b_\varepsilon(\rho   ||   \gamma)=e^{-\Dh^\varepsilon(\rho   ||   \gamma)}$.
\end{definition}
\noindent Further details appear in~\cite{DupuisKFRR12,YungerHalpernR14}.

%
%
\subsection{Quantification of one-shot work quantities}

$\Wext$ can be calculated, and $\Wform$ can be bounded, in terms of $\Dh^\varepsilon$.

\begin{theorem} 
\label{theorem:OneShotWork}
Let $R   :=   (r,  H,  X_1, \ldots, X_j)$ denote any quasiclassical state in any energy-representation thermodynamic resource theory $\mathcal{T}^{\beta, p_1, \ldots, p_j}$, and
let $G_R   :=   (g_R,  H,  X_1, \ldots, X_j)$ denote the corresponding equilibrium state.
The most work $\Wext(R)$ $\varepsilon$-extractable from $R$ is
\begin{align}  \label{eq:WExt}
   \Wext(R)   =   \frac{1}{\beta}   \Dh^{\varepsilon}(r   ||   g_R)
\end{align}
for all $\varepsilon  \in  [0, 1]$. The least work needed to form any $\tilde{R}  \approx_\varepsilon  R$ satisfies
\begin{align}  \label{eq:WForm}
 \max_{\delta   \in (0,   1-\varepsilon]}  \left[
     \frac{1}{\beta}   \Dh^{1 - \varepsilon - \delta}(r   ||   g_R)
     -   \frac{1}{\beta}   \log \left( \frac{1}{\delta} \right)   \right]
   \leq \Wform(R)
   \leq \frac{1}{\beta}   \Dh^{1 - \varepsilon}(r   ||   g_R)
        -  \frac{1}{\beta}   \log   \left( \frac{1-\varepsilon}{\varepsilon} \right)
\end{align}
for all $\varepsilon \in [0, 1]$.
\end{theorem}

\begin{proof}
The proof immediately generalizes the proof of~\cite[Theorem 5]{YungerHalpernR14}. The latter theorem is the manifestation, in grand-potential theories, of Theorem~\ref{theorem:OneShotWork} above. The grand-potential proof relies on equimajorization and on hypothesis tests between $r$ and $g_R$, which are well-defined in $\mathcal{T}^{\beta, p_1, \ldots, p_j}$. Extensive-variable operators ($X_0 = H$ and $X_1 = N$) do not appear in the grand-potential proof directly. The $X_i$'s affect the proof insofar as $g_R$ depends on them. Hence the proof of~\cite[Theorem 5]{YungerHalpernR14} can be restated in $\mathcal{T}^{\beta, p_1, \ldots, p_j}$ under the assumption that $g_R$ has the form in Eq.~\eqref{eq:FreeState}.
\end{proof}

%
%
%
%
\section{Thermodynamic limit, generalization of Szil\'{a}rd's engine and Landauer erasure}
\label{section:ThermLimit}

Consider distilling work from, or creating, many copies of a quasiclassical state $R$ in any thermodynamic resource theory $\mathcal{T}^{\beta, p_1, \ldots, p_j}$  other than the entropy theory:
\begin{equation*}
   R^{\otimes n}  :=  
   \left(   r^{\otimes n},  \sum_{\mu =1}^n H_\mu,   \sum_{\mu =1}^n X_{1_\mu},  
             \ldots,   \sum_{\mu=1}^n X_{j_\mu}   \right).
\end{equation*}
The \emph{asymptotic}, or thermodynamic, limit is defined by $n \to \infty$.
We shall see that $\Wext( R^{\otimes n} )$ and $\Wform( R^{ \otimes n } )$ converge, as $n \to \infty$, to differences between functions reminiscent of free energies. I will contrast the energy representation of most thermodynamic resource theories with the entropy representation of the entropy theory; will derive the optimal asymptotic rate of interconversion between states; and will show how the interconvertibility of types of thermodynamic resourcefulness generalizes Szil\'{a}rd work extraction and Landauer erasure.  By generalizing~\cite{YungerHalpernR14}, I will show that $\Wext( R^{\otimes n} )$ can differ from $\Wform( R^{\otimes n} )$.

The asymptotic limit follows from the Asymptotic Equipartition Theorem (AEP)~\cite{Hiai_91_Proper}
\begin{equation}   \label{eq:AEP}
   \lim_{n \to \infty}   \frac{1}{n}
   \Dh^\varepsilon ( r^{\otimes n} || g_R^{\otimes n} )
   =   D(r || g_R)
   \quad \forall \varepsilon \in (0, 1),
\end{equation}
wherein the relative entropy is 
$D( r || g_R)  :=  \sum_\alpha [ r_\alpha  (\log r_\alpha  -  \log g_\alpha) ]$.
AEPs have been used to derive the asymptotic limits of Helmholtz theories~\cite{FundLimits2}, the entropy theory~\cite{GourMNSYH13}, and grand-potential theories~\cite{YungerHalpernR14}. 
Applying Eq.~\eqref{eq:AEP} to Eq.~\eqref{eq:WExt} and to Ineqs.~\eqref{eq:WForm} yields
\begin{equation}
   \lim_{n \to \infty}  \Wext( R^{\otimes n} )
   =   \lim_{n \to \infty}  \Wform( R^{\otimes n} )
   =   \frac{1}{\beta}  D( r || g_R )
   \quad \forall \varepsilon \in (0, 1).
\end{equation}
Substituting in the definition of $D$ will clarify this result. I will define logarithms as base-$e$, denote eigenvalue $\alpha$ of $X_i$ by $x_{i_{ \alpha }}$, and denote the partition function by $Z$:
\begin{align}
   \frac{1}{\beta}  D(r || g_R)
   &  =   \frac{1}{\beta}   \sum_\alpha
             r_\alpha   \left[   \log (r_\alpha)
             -   \log \left(  \frac{ e^{ - \frac{1}{ k_B } 
                                           \left(  F_0 x_{0_{ \alpha }}    +  F_1 x_{1_{ \alpha }}    
                                           + \ldots + F_j x_{j_{ \alpha }} \right) } }{
                           Z }   \right)   \right]  \\
   & = \frac{1}{\beta}   \sum_\alpha   r_\alpha  \log (r_\alpha)  
           +   \sum_\alpha  r_\alpha  ( E_{\alpha}  -  p_1 x_{1_{\alpha}}  -  \ldots  
                    -  p_j x_{j_{\alpha}} )  
           +   \frac{1}{\beta}  \log Z  \\
   & =  \label{eq:AsympRate}
           \langle H \rangle_r   -   T \: \kB S(r)   -   p_1  \langle X_1 \rangle_r   
            -   \ldots   -   p_j \langle X_j \rangle_r  
           +   \kB T  \log Z.
\end{align}
The dimensional factor $\kB$ converts the Shannon entropy $S(r)  :=  -\sum_i r_i \log (r_i)$ into statistical-mechanical entropy~\cite{Belof09}, and $\langle . \rangle_r$ denotes an expectation value relative to the distribution $r$. 

In the asymptotic limit, the average work extractable from one copy of $R$, and the average work needed to create one copy, with faulty protocols approach the free energy that lends its name to $\mathcal{T}^{\beta, p_1, \ldots, p_j}$.
In Helmholtz theories, Eq.~\eqref{eq:AsympRate} has the form of $F  :=  E - TS$; in grand-potential theories, the form of $\Phi  :=  E - TS - \mu N$. Equation~\eqref{eq:AsympRate} echoes Eq.~\eqref{eq:ThermoWork}, the thermodynamic equality for the average work extractable or required during a quasistatic transformation: $W  =  \Delta$(Free energy). Resource-theory work quantities not only shed new light on single-state transformations, but also converge, in the asymptotic limit, to expressions one might expect from thermodynamics.

That Eq.~\eqref{eq:AsympRate} contains an energy contrasts with the analogous result in the entropy theory. The resource in the entropy theory has been called \emph{information}~\cite{HHOShort,HHOLong} and \emph{nonuniformity}~\cite{GourMNSYH13}. The average resourcefulness extractable from, and the average resourcefulness needed to create, one copy of $R$ has been shown to equal
\begin{equation}   \label{eq:NonunAsymp}
   D(r  ||  g_R)   =   \log  d   -   S(r)
\end{equation}
in the asymptotic limit~\cite{HHOShort,HHOLong}. $d$ denotes the number of elements in $r$, and $g_R  =  ( \underbrace{ \frac{1}{d}, \ldots, \frac{1}{d} }_d )$ denotes a microcanonical ensemble.  Equation~\eqref{eq:AsympRate} has dimensions of energy, whereas Eq.~\eqref{eq:NonunAsymp} is dimensionless like entropy. The entropy-theory Eq.~\eqref{eq:NonunAsymp} is in the entropy representation, whereas the Eq.~\eqref{eq:AsympRate} that characterizes other thermodynamic resource theories is in the energy representation (in the sense of Sec.~\ref{section:Background}.\ref{section:FundRelation}).

Equation~\eqref{eq:AsympRate} implies the \emph{optimal asymptotic rate of conversion}. Let $m_n$ denote the greatest number of copies of $S$ that equilibrating operations can generate from $n$ copies of $R$. The \emph{optimal rate of conversion from $n$ copies} is
\begin{equation*}
   \mathcal{R}_n (R  \xmapsto{equil.}   S)  =   \frac{ m_n }{n}.
\end{equation*}
In the asymptotic limit, $\mathcal{R}_n$ approaches the \emph{optimal asymptotic conversion rate}:
\begin{equation*}
   \lim_{n \to \infty}   \mathcal{R}_n( R  \xmapsto{equil.}   S )   
   =   \mathcal{R}_\infty( R  \xmapsto{equil.}   S ).
\end{equation*}
In $\mathcal{T}^{(\beta), p_1, \ldots, p_j}$,
\begin{equation}  \label{eq:AsympConvert}
   \mathcal{R}_\infty( R  \xmapsto{equil.}   S )   
   =   \frac{ D(r || g_R) }{ D( s || g_S ) }.
\end{equation}
During the optimal conversion protocol in an energy-representation theory, $D(r || g_R)$ units of work are extracted per copy of $R$, and one copy of $S$ is generated per $D( s || g_S )$ units of work. Equation~\eqref{eq:AsympConvert} has been derived for Helmholtz theories~\cite{BrandaoHORS13,FundLimits2}, grand-potential theories~\cite{YungerHalpernR14}, and the entropy theory~\cite{HHOShort,HHOLong,GourMNSYH13}.

Equation~\eqref{eq:AsympConvert} implies that, in the asymptotic limit, all quasiclassical nonequilibrium states in $\mathcal{T}^{(\beta), p_1, \ldots, p_j}$ are \emph{reversibly interconvertible}. From enough copies of any nonequilibrium $R$, equilibrating operations can generate copies of any $S$. The optimal $R$-to-$S$ rate equals the inverse of the optimal $S$-to-$R$ rate. 
This interconvertibility might surprise us. The resourcefulness of $R$ might manifest in one form---as information, chemical energy, magnetic energy, etc.---while the resourcefulness of $S$ might manifest in another. That thermodynamic resourcefulness of each sort can transform into resourcefulness of every other is not \emph{a priori} obvious. After all, though LOCC can asymptotically interconvert all pure bipartite entangled states, LOCC cannot interconvert all pure tripartite entangled states~\cite{LindenPSW99,VDC00,WuZ01}. Equation~\eqref{eq:AsympConvert} implies that thermodynamic resourcefulness resembles bipartite, more than tripartite, entanglement.~\footnote{
Interconversions amongst information, energy, and angular momentum were analyzed outside of resource theories in~\cite{VaccaroB11,BarnettV13}.
}

This transformability generalizes Szil\'{a}rd's engine and Landauer erasure. Szil\'{a}rd's engine converts information into work~\cite{Szilard29}. The classical engine consists of a particle known to occupy a box's left-hand side. Consider sliding a partition through the box's center, attaching a weight to the partition, and coupling the particle to a heat bath. The particle, modeled by an expanding ideal gas, pushes the partition to the box's right-hand side. The weight rises, its gravitational potential energy compensating for the loss of information about the particle's location. Landauer erasure, resulting from reversing Szil\'{a}rd's engine, converts work into information~\cite{Landauer61}. In general thermodynamic resource theories, resourcefulness can be converted not only between information and gravitational potential energy, but amongst all the DOFs represented by state operators. Value can be encoded in particle number, magnetic moments, etc.

Though the large-$n$ limit of $\Wext(R^{\otimes n})$ equals that of 
$\Wform(R^{\otimes n})$, the rates at which these quantities approach the limit differ. Information theory has been used to show that, in grand-potential theories, 
\begin{equation}
   \Wext(R^{\otimes n})  \approx  \frac{1}{\beta} [ n \, D(r  ||  g_R)   -   O( \sqrt{n} )]
\end{equation}
and
\begin{equation}
   \Wform(R^{\otimes n})  \approx  \frac{1}{\beta} [ n \, D(r  ||  g_R)   +   O( \sqrt{n} )]
\end{equation}
as $n$ grows large~\cite{YungerHalpernR14}. The proof does not depend on the extensive-variable state operators $X_{i = 0, 1, \ldots, j}$ explicitly. The proof depends on the $X_i$'s only implicitly, insofar as $g_R$ depends on the $X_i$. As in the proof of Theorem~\ref{theorem:Quasiorder}, the grand-potential proof can be restated in $\mathcal{T}^{\beta, p_1, \ldots, p_j}$, under the assumption that $g_R$ has the form in Eq.~\eqref{eq:FreeState}. Creating a state $R$ costs $O ( \sqrt{n} )$ units of work more than can be extracted from $R$. We pay more than we gain, as expected from the Second Law of Thermodynamics.

%
%
%
%
\section{Discussion}
\label{section:Discussion}

This paper has motivated and detailed the generalization of thermodynamic resource theories beyond heat baths. Traditional thermodynamics suggested a correspondence between one family of thermodynamic resource theories and each type of interaction (equivalently, each set of natural variables, each thermodynamic potential, and each type of bath or external field). Generalizing thermodynamic resource theories opens diverse realistic systems---possibly such as polymers, electrochemical batteries, and magnets---to modeling. This opening is hoped to facilitate experimental tests of one-shot statistical mechanics.

The generalization opens questions to investigation. First, operators $X_i$ and $X_j$ that represent external variables might fail to commute. What implications does this nonclassicality have for thermodynamic resource theories?
Second, the number operator $N$ enables one copy of a state $R$ to consist of a macroscopic number of particles, inviting a reconsideration of what ``one-shot'' means. The reconsideration might facilitate experimental tests of one-shot theory, as might the mathematical similarity amongst the families of thermodynamic resource theories. Finally, one might eliminate the asymmetry between free-energy theories (like Helmholtz theories~\cite{BrandaoHORS13}), which have been cast in terms of the energy representation of thermodynamics, and the entropy theory~\cite{HHOShort,GourMNSYH13}, which has been cast in terms of the entropy representation.

These challenges might offer novel insights into thermodynamics. A more-straightforward task is to ensure that all Helmholtz-theory results generalize to arbitrary thermodynamic resource theories $\mathcal{T}^{\beta, p_1, \ldots, p_j}$. Theorems~\ref{theorem:FreeStates},~\ref{theorem:Quasiorder}, and~\ref{theorem:LorenzCurves} generalize theorems about Helmholtz and grand-potential theories. That other Helmholtz-theory results, including ones about catalysis~\cite{BrandaoHNOW14} and coherences~\cite{BrandaoHORS13,FaistOR14,LostaglioJR14,LostaglioKJR15}, generalize merit checking.

%
%
\subsection{Noncommutation}
\label{section:Quantum}

Many Helmholtz-theory results concern states whose density operators $\rho$ commute with their Hamiltonians $H$~\cite{Janzing00,FundLimits2,YungerHalpernR14}. Noncommutation of $\rho$ and $H$ has been explored recently~\cite{BrandaoHORS13,FaistOR14,LostaglioJR14,LostaglioKJR15}. $\rho$ and $H$ are the only operators that can fail to commute, because they alone define states in Helmholtz theories. In a general thermodynamic resource theory $\mathcal{T}^{\beta, p_1, \ldots, p_j}$, more operators can fail to commute: $\rho$, $H$, and the other $X_i$ that represent extensive variables. Jaynes considered such noncommutation (outside of resource theories)~\cite{JaynesII57}.
Noncommutation impacts free states, free operations, and interpretations of thermodynamic resource theories.

Since the initial arXiv release of this paper, in 2014,
this noncommutation has studied~\cite{Lostaglio_17_Thermodynamic,Guryanova_16_Thermodynamics,YH_16_Microcanonical}.
Further works have build on this foundation
(e.g.,~\cite{Ito_16_Optimal,NathBera_17_Thermodynamics,MurPetit_17_Generalised}).
The present paper has provided motivation for these later works.
Similar motivation was developed independently
in~\cite{Lostaglio_14_Masters}.

Let $R := (\rho, H, X_1, \ldots, X_j)$ denote a state whose $[H, X_1]  \neq  0$. Consider attempting to associate $R$ with a free state $G_R$. The density operator $\gamma_R$ of $G_R$ could not be derived as in Appendix~\ref{section:FreeStateApp}. 
$\gamma_R$ is shown, under the assumption that the $X_i$ commute, to depend on $X_i$ as
\begin{equation}   \label{eq:ProductSum}
   \gamma_R \propto
   \prod_{i = 0}^j   e^{ -\frac{1}{ \kB }  F_i X_i }
   =   e^{  - \frac{1}{ \kB }    \sum_{i = 0}^j  F_i  X_i  }
\end{equation}
[Eq.~\eqref{eq:FreeState}]. If the $X_i$ fail to commute, the Baker-Campbell-Hausdorff Formula could not straightforwardly justify the equality in Eq.~\eqref{eq:ProductSum}.

Noncommutation can restrict the unitaries that can evolve $R$ for free. In a Helmholtz theory $\mathcal{T}^{\beta}$, nontrivial functions $U = f(H)$ of $H$ can evolve states 
$R  :=  (\rho, H)$. 
In a non-Helmholtz theory $\mathcal{T}^{\beta, p_1, \ldots, p_j}$, noncommutation of $H$ with any $X_{i \neq 0}$ can prevent states from evolving for free under all nontrivial $U = f(H)$. Such $U$'s could fail to preserve $X_i$, violating Eq.~\eqref{eq:ConstrainU}. In extreme situations, free unitaries may be restricted to $U(\phi)  = e^{i \phi}  \id$ for $\phi \in \mathbbm{R}$. 

For example, let 
\mbox{$R := (\rho, H, X_1, X_2)$} denote a state whose $[X_1, X_2]  \neq  0$ in
$\mathcal{T}^{\beta, p_1, p_2}$. A simple Hamiltonian might have the form 
$H  =  p_1 X_1  +  p_2  X_2$.\footnote{
This form of $H$ is inspired by Fundamental Relations such as the $E = \mathbf{B} \cdot \mathbf{m}$ of a system whose magnetization is $\mathbf{m}$ and that occupies an external magnetic field $\mathbf{B}$.}
As $H$ does not commute with $X_1$ or $X_2$, neither does $U(t)  =  e^{ - \frac{i}{\hbar} H t }$. Time evolution does not manifest as a free unitary. To understand the physical significance of this surprising mathematical conclusion, one might apply insights from quantum mechanical SO coupling, in which the orbital angular momentum $L$ and the spin angular momentum $S$ are not conserved separately.

%
%
%
%
\subsection{Meaning of ``one-shot,'' experimental tests}
\label{section:OneShot}

The entropy, Helmholtz, and grand-potential theories have been portrayed as alternatives to traditional thermodynamics, which describes on the order of $10^{24}$ particles~\cite{FundLimits2,BrandaoHNOW14,YungerHalpernR14,LostaglioJR14}. In these theories (more generally, in one-shot statistical mechanics), one-shot information theory is applied to single copies of a state. The work extractable from, and the work cost of creating, $n$ copies of a state converge, as $n \to \infty$, to expressions reminiscent of differences between thermodynamic free energies [Eq.~\eqref{eq:AsympRate}]. 

The significance of $n$ merits reevaluation in the light of $N$. $N$ appears to belong in thermodynamic resource theories: $N$ resembles $E$ in conventional thermodynamics, and $E$ manifests in resource theories (as $H$).
Particle number resembles energy because both are represented by extensive variables. As Jaynes says, ``the energy plays a preferred role among all dynamical quantities because it is conserved \ldots however, \ldots all measurable quantities may be treated on the same basis, subject to certain precautions''~\cite{JaynesI57}.
Yet $N$ can have arbitrarily large eigenvalues. In Helmholtz theories, $N$ remains ``behind the scenes'' due to the fixed-eigensubspace condition [Proposition~\ref{proposition:Fixed}]. Yet $\supp(\rho)$ can occupy an eigensubspace of $N$ associated with many particles. One copy of a state $R  :=  (\rho,  H)$ can correspond to $10^{24}$ particles, to a conventional thermodynamic system. Yet one copy of $R$ has been called ``one-shot'' and has been contrasted with conventional thermodynamics.

Magnets illustrate this seeming contradiction. Consider a magnet that consists of about $10^{24}$ spins, whose total magnetic moment is $\mathbf{m}$, and that occupies the state $R$. If $T$ and the external magnetic field $\mathbf{B}$ remain constant, a resource theory for $E[ T, \mathbf{B} ]$ models the distillation of work from $R$. By Eq.~\eqref{eq:WExt}, the work extractable from one system of approximately $10^{24}$ spins is the one-shot quantity $\Wext(R)$.

This paradox might facilitate experimental tests of one-shot statistical mechanics. Whether $\supp(\rho)$ occupies the $N = 1$ eigensubspace or the $N = 10^{24}$ eigensubspace does not affect resource-theory calculations. $N=1$ states interest one-shot theorists, whereas mathematically equivalent $N = 10^{24}$ states can be controlled more easily in laboratories.
To test one-shot results such as Eqs.~\eqref{eq:WExt} and~\eqref{eq:WForm}, one might perhaps use systems of $10^{24}$ particles. Ascertaining the form of a $\rho$ that characterizes a large system poses a practical challenge.

Like $N$, the mathematical equivalence of families of thermodynamic resource theories might facilitate experimental tests. Having been theoretically developed extensively, Helmholtz theories merit testing. Helmholtz theories are equivalent, up to a Legendre transform, to Gibbs-free-energy theories. Chemists apply the Gibbs free energy \mbox{$G  :=  E[T, p]  :=  E  -  TS  +  pV$} often, as to room-temperature, atmospheric-pressure chemical reactions. Whether common experiments could test one-shot statistical mechanics remains to be explored. 
Polymers offer another possible platform. The mechanical force $f$ applied to a length-$L$ polymer appears in the free energy
$F  =  - pV + \mu N + fL$
\cite{Alberty01}.
This $F$ corresponds to a family of thermodynamic resource theories, and polymers have been strained to test fluctuation relations experimentally~\cite{MossaMFHR09,ManosasMFHR09,AlemanyR10}. Generalizing thermodynamic resource theories expands the theories' potential for modeling real physical systems.

%
%
%
%
\subsection{Further generalization}

Conventional thermodynamics can be formulated equivalently in energy and entropy representations (Sec.~\ref{section:Background}). In most of the resource theories in this paper, resourcefulness is quantified in terms of work, and average work yields and work costs converge to free energies in the asymptotic limit. These theories are cast in terms of the energy representation. In the entropy theory (or the resource theory for nonuniformity, information, or informational nonequilibrium~\cite{HHOShort,HHOLong,GourMNSYH13}), a $d$-dimensional state $r$ represents (as has been argued in this paper) a closed isolated system. The average resourcefulness extractable from, and the average resourcefulness cost of creating, $r$ converge to the entropic quantity $\log d  -  S(r)$~\cite{HHOShort}. The entropy theory is cast in terms of the entropy representation.

That an energy-representation resource theory models closed isolated systems seems a reasonable expectation. That entropy-representation theories are equivalent to Helmholtz and grand-potential theories seems a reasonable expectation. Realizing or discounting those expectations may shed new light on one-shot statistical mechanics.

%
%
%
%
\section*{Acknowledgements}
The author is grateful for conversations with Ian~Durham, Tobias~Fritz, David~Jennings, Matteo~Lostaglio, Iman~Marvian, Evgeny~Mozgunov, Markus~P.~M\"{u}ller, Joseph~M.~Renes, Brian~Space, and Rob~Spekkens. This research was supported by a Virginia Gilloon Fellowship, an IQIM Fellowship, NSF grant PHY-0803371, and the Perimeter Institute for Theoretical Physics. The Institute for Quantum Information and Matter (IQIM) is an NSF Physics Frontiers Center supported by the Gordon and Betty Moore Foundation. Research at the Perimeter Institute is supported by the Government of Canada through Industry Canada and by the Province of Ontario through the Ministry of Research and Innovation.

%
%
%
%

\section{Appendices}
\begin{appendices}

%
%
%
%
\section{Volume as an operator?}
\label{section:Volume}

In Sec.~\ref{section:States}, extensive thermodynamic variables are associated with operators. The volume $V$ is such a variable.  Jaynes represents volume with an operator in~\cite{JaynesII57}.
But what is a volume operator? It might be understood similarly to the position operator $x$. The position operator is often explained in terms of a photon that scatters off a particle (whose position is to be measured) and into a detector. From this story, we extrapolate to the idea of a general position operator $x$. A similar story might be told about $V$.

Consider a gas in a box that consists of five rigid walls and one moveable partition. The partition's position governs the gas's volume. Since the partition is a quantum object, its position fluctuates and merits modeling with an operator $x$. So, by extension, does the box's volume.

Like $x$ and unlike $N$, $V$ has a continuous spectrum. A rigorous incorporation of this continuity into Sections~\ref{section:DefineTheories}-\ref{section:ThermLimit} is expected to follow from a discretization of space and a limit as the discretization vanishes~\cite{RenesYHM13}. 
Space may simply be discretized for the purposes of this paper, as in~\cite{SkrzypczykSP13Extract}.

Sufficiently large numbers $n_i$ of particles cannot fit in sufficiently small volumes (if $\mathcal{S}$ is not a relativistic system such as a black hole). For example, the macroscopic $n_i = 10^{24}$ seems incompatible with the atomic-scale $v_j = 10^{-30} \: {\rm m}^3$. How can $N$ and $V$ share an eigenstate
$\ket{ n_i, v_j }  =  \ket{ 10^{24},10^{-30} \: {\rm m}^3 }$? This eigenstate must correspond to an energy so enormous that $\mathcal{S}$ effectively cannot occupy $\ket{ n_i, v_j }$. Systems are assumed not to occupy states associated with energies above some cutoff.

%
%
%
%
\section{Derivation of the form of free states}
\label{section:FreeStateApp}

Theorem~\ref{theorem:FreeStates} can be proved as follows.

\begin{proof}
Theorem~\ref{theorem:FreeStates} trivially generalizes~\cite[Theorem 1]{YungerHalpernR14}. In~\cite{YungerHalpernR14}, the grand-canonical form of the $g$'s in grand-potential theories was derived in three steps. Let us sketch these steps first.

First, the derivation of the microcanonical form $( \underbrace{ \frac{1}{d}, \ldots, \frac{1}{d} }_d )$ of the entropy theory's free states is reviewed. Horodecki \emph{et al.} derive the form via a proof by contradiction~\cite{HHOLong}. Suppose that some nonuniform state $u_0$ were free. An agent could generate some large number $n$ of copies of $u_0$. Using free unitaries, the agent could Shannon-compress the mixedness in $(u_0)^{\otimes n}$, separating the mixedness from a state arbitrarily close to being pure. Able to create purity and mixedness, the agent could create states other than $u_0$ for free. 

Second, the free states in Helmholtz theories are shown to be canonical ensembles: 
$g_i  =  e^{- \beta E_i} / Z$. The weight of $g$ on each energy subspace $S_E$ is shown to be distributed uniformly across the energy levels that span $S_E$ because of Step 1. Then, the ratio $g(E + \Delta) / g(E)$ of two weights is shown to depend only on the gap $\Delta$ between the levels. The ratio is shown to vary as $e^{-\beta \Delta}$, and all gaps in all free states are shown to correspond to the same $\beta$.\footnote{
An alternative resource-theory derivation of the canonical ensemble appears in~\cite{BrandaoHNOW14}.}

Finally, the free states $G := (g, H, N)$ in grand-potential theories are shown to be grand canonical ensembles. If $N$ is totally degenerate, $G$ represents a bath whose particle number remains constant, like the free states in a Helmholtz theory. Hence $g_i$ varies with $E_i$ as $e^{-\beta E_i}$. By an analogous argument, $g_i$ varies with $n_i$ as $e^{\beta \mu n_i}$. Hence 
$g_i = e^{- \beta (E_i - \mu n_i)} / Z$.
The canonical-ensemble derivation generalizes immediately to 
$\mathcal{T}^{\beta, p_1, \ldots, p_j}$. The argument applied to $N$ applies to each $X_i$, so Eq.~\eqref{eq:FreeState} follows from 
\begin{align}
   g_\alpha  
   & =  e^{ - \beta (E_\alpha  -  p_1 x_{1_\alpha}  -  \ldots  -  p_j x_{j_\alpha}) }  / Z   \\
   & =  e^{ - \frac{1}{\kB} \left(  
                       \frac{1}{T}  x_{0_\alpha}   +   \frac{p_1}{T} x_{1_\alpha}   +   \ldots   
                       +   \frac{p_j}{T} x_{j_\alpha}   \right) } / Z   \\
   & =  e^{ - \frac{1}{\kB} \left(  
                       F_0 x_{0_\alpha}   +   \ldots   +   F_j x_{j_\alpha}   \right) } / Z.
\end{align}
\end{proof}

\end{appendices}

%
%
%
\bibliography{NoneqRefs}

\begin{thebibliography}{62}%
\makeatletter
\providecommand \@ifxundefined [1]{%
 \@ifx{#1\undefined}
}%
\providecommand \@ifnum [1]{%
 \ifnum #1\expandafter \@firstoftwo
 \else \expandafter \@secondoftwo
 \fi
}%
\providecommand \@ifx [1]{%
 \ifx #1\expandafter \@firstoftwo
 \else \expandafter \@secondoftwo
 \fi
}%
\providecommand \natexlab [1]{#1}%
\providecommand \enquote  [1]{``#1''}%
\providecommand \bibnamefont  [1]{#1}%
\providecommand \bibfnamefont [1]{#1}%
\providecommand \citenamefont [1]{#1}%
\providecommand \href@noop [0]{\@secondoftwo}%
\providecommand \href [0]{\begingroup \@sanitize@url \@href}%
\providecommand \@href[1]{\@@startlink{#1}\@@href}%
\providecommand \@@href[1]{\endgroup#1\@@endlink}%
\providecommand \@sanitize@url [0]{\catcode `\\12\catcode `\$12\catcode
  `\&12\catcode `\#12\catcode `\^12\catcode `\_12\catcode `\%12\relax}%
\providecommand \@@startlink[1]{}%
\providecommand \@@endlink[0]{}%
\providecommand \url  [0]{\begingroup\@sanitize@url \@url }%
\providecommand \@url [1]{\endgroup\@href {#1}{\urlprefix }}%
\providecommand \urlprefix  [0]{URL }%
\providecommand \Eprint [0]{\href }%
\providecommand \doibase [0]{http://dx.doi.org/}%
\providecommand \selectlanguage [0]{\@gobble}%
\providecommand \bibinfo  [0]{\@secondoftwo}%
\providecommand \bibfield  [0]{\@secondoftwo}%
\providecommand \translation [1]{[#1]}%
\providecommand \BibitemOpen [0]{}%
\providecommand \bibitemStop [0]{}%
\providecommand \bibitemNoStop [0]{.\EOS\space}%
\providecommand \EOS [0]{\spacefactor3000\relax}%
\providecommand \BibitemShut  [1]{\csname bibitem#1\endcsname}%
\let\auto@bib@innerbib\@empty
\bibitem [{\citenamefont {Lacoste}\ \emph {et~al.}(2008)\citenamefont
  {Lacoste}, \citenamefont {Lau},\ and\ \citenamefont {Mallick}}]{LacosteLM08}%
  \BibitemOpen
  \bibfield  {author} {\bibinfo {author} {\bibfnamefont {D.}~\bibnamefont
  {Lacoste}}, \bibinfo {author} {\bibfnamefont {A.~W.}\ \bibnamefont {Lau}}, \
  and\ \bibinfo {author} {\bibfnamefont {K.}~\bibnamefont {Mallick}},\ }\href
  {\doibase 10.1103/PhysRevE.78.011915} {\bibfield  {journal} {\bibinfo
  {journal} {Phys. Rev. E}\ }\textbf {\bibinfo {volume} {78}},\ \bibinfo
  {pages} {011915} (\bibinfo {year} {2008})}\BibitemShut {NoStop}%
\bibitem [{\citenamefont {Serreli}\ \emph {et~al.}(2007)\citenamefont
  {Serreli}, \citenamefont {Lee},\ and\ \citenamefont {Leigh}}]{SerreliLKL07}%
  \BibitemOpen
  \bibfield  {author} {\bibinfo {author} {\bibfnamefont {V.}~\bibnamefont
  {Serreli}}, \bibinfo {author} {\bibfnamefont {E.~R.}\ \bibnamefont {Lee},
  \bibfnamefont {C.-F.~andKay}}, \ and\ \bibinfo {author} {\bibfnamefont
  {D.~A.}\ \bibnamefont {Leigh}},\ }\href {\doibase doi:10.1038/nature05452}
  {\bibfield  {journal} {\bibinfo  {journal} {Nature}\ }\textbf {\bibinfo
  {volume} {445}},\ \bibinfo {pages} {523} (\bibinfo {year}
  {2007})}\BibitemShut {NoStop}%
\bibitem [{\citenamefont {Bustamante}\ \emph {et~al.}(2005)\citenamefont
  {Bustamante}, \citenamefont {Liphardt},\ and\ \citenamefont
  {Ritort}}]{BustamanteLP05}%
  \BibitemOpen
  \bibfield  {author} {\bibinfo {author} {\bibfnamefont {C.}~\bibnamefont
  {Bustamante}}, \bibinfo {author} {\bibfnamefont {J.}~\bibnamefont
  {Liphardt}}, \ and\ \bibinfo {author} {\bibfnamefont {F.}~\bibnamefont
  {Ritort}},\ }\href {\doibase 10.1063/1.2012462} {\bibfield  {journal}
  {\bibinfo  {journal} {Physics Today}\ }\textbf {\bibinfo {volume} {58}},\
  \bibinfo {pages} {43} (\bibinfo {year} {2005})}\BibitemShut {NoStop}%
\bibitem [{\citenamefont {Collin}\ \emph {et~al.}(2005)\citenamefont {Collin},
  \citenamefont {Ritort}, \citenamefont {Jarzynski}, \citenamefont {Smith},
  \citenamefont {Tinoco},\ and\ \citenamefont {Bustamante}}]{CollinRJSTB05}%
  \BibitemOpen
  \bibfield  {author} {\bibinfo {author} {\bibfnamefont {D.}~\bibnamefont
  {Collin}}, \bibinfo {author} {\bibfnamefont {F.}~\bibnamefont {Ritort}},
  \bibinfo {author} {\bibfnamefont {C.}~\bibnamefont {Jarzynski}}, \bibinfo
  {author} {\bibfnamefont {S.~B.}\ \bibnamefont {Smith}}, \bibinfo {author}
  {\bibfnamefont {I.}~\bibnamefont {Tinoco}}, \ and\ \bibinfo {author}
  {\bibfnamefont {C.}~\bibnamefont {Bustamante}},\ }\href {\doibase
  10.1038/nature04061} {\bibfield  {journal} {\bibinfo  {journal} {Nature}\
  }\textbf {\bibinfo {volume} {437}},\ \bibinfo {pages} {231} (\bibinfo {year}
  {2005})}\BibitemShut {NoStop}%
\bibitem [{\citenamefont {Liphardt}\ \emph {et~al.}(2002)\citenamefont
  {Liphardt}, \citenamefont {Dumont}, \citenamefont {Smith}, \citenamefont
  {Tinoco},\ and\ \citenamefont {Bustamante}}]{LiphardtDSTB02}%
  \BibitemOpen
  \bibfield  {author} {\bibinfo {author} {\bibfnamefont {J.}~\bibnamefont
  {Liphardt}}, \bibinfo {author} {\bibfnamefont {S.}~\bibnamefont {Dumont}},
  \bibinfo {author} {\bibfnamefont {S.~B.}\ \bibnamefont {Smith}}, \bibinfo
  {author} {\bibfnamefont {I.}~\bibnamefont {Tinoco}}, \ and\ \bibinfo {author}
  {\bibfnamefont {C.}~\bibnamefont {Bustamante}},\ }\href {\doibase
  10.1126/science.1071152} {\bibfield  {journal} {\bibinfo  {journal} {Science
  (New York, N.Y.)}\ }\textbf {\bibinfo {volume} {296}},\ \bibinfo {pages}
  {1832} (\bibinfo {year} {2002})}\BibitemShut {NoStop}%
\bibitem [{\citenamefont {Alemany}\ and\ \citenamefont
  {Ritort}(2010)}]{AlemanyR10}%
  \BibitemOpen
  \bibfield  {author} {\bibinfo {author} {\bibfnamefont {A.}~\bibnamefont
  {Alemany}}\ and\ \bibinfo {author} {\bibfnamefont {F.}~\bibnamefont
  {Ritort}},\ }\href@noop {} {\bibfield  {journal} {\bibinfo  {journal}
  {Europhysics News}\ }\textbf {\bibinfo {volume} {41}},\ \bibinfo {pages} {27}
  (\bibinfo {year} {2010})}\BibitemShut {NoStop}%
\bibitem [{\citenamefont {Cheng}\ \emph {et~al.}(2012)\citenamefont {Cheng},
  \citenamefont {Sreelatha}, \citenamefont {Hou}, \citenamefont {Efremov},
  \citenamefont {Liu}, \citenamefont {van~der Maarel},\ and\ \citenamefont
  {Wang}}]{ChengSHEL12}%
  \BibitemOpen
  \bibfield  {author} {\bibinfo {author} {\bibfnamefont {J.}~\bibnamefont
  {Cheng}}, \bibinfo {author} {\bibfnamefont {S.}~\bibnamefont {Sreelatha}},
  \bibinfo {author} {\bibfnamefont {R.}~\bibnamefont {Hou}}, \bibinfo {author}
  {\bibfnamefont {A.}~\bibnamefont {Efremov}}, \bibinfo {author} {\bibfnamefont
  {R.}~\bibnamefont {Liu}}, \bibinfo {author} {\bibfnamefont {J.~R.~C.}\
  \bibnamefont {van~der Maarel}}, \ and\ \bibinfo {author} {\bibfnamefont
  {Z.}~\bibnamefont {Wang}},\ }\href {\doibase 10.1103/PhysRevLett.109.238104}
  {\bibfield  {journal} {\bibinfo  {journal} {Phys. Rev. Lett.}\ }\textbf
  {\bibinfo {volume} {109}},\ \bibinfo {pages} {238104} (\bibinfo {year}
  {2012})}\BibitemShut {NoStop}%
\bibitem [{\citenamefont {Janzing}\ \emph {et~al.}(2000)\citenamefont
  {Janzing}, \citenamefont {Wocjan}, \citenamefont {Zeier}, \citenamefont
  {Geiss},\ and\ \citenamefont {Beth}}]{Janzing00}%
  \BibitemOpen
  \bibfield  {author} {\bibinfo {author} {\bibfnamefont {D.}~\bibnamefont
  {Janzing}}, \bibinfo {author} {\bibfnamefont {P.}~\bibnamefont {Wocjan}},
  \bibinfo {author} {\bibfnamefont {R.}~\bibnamefont {Zeier}}, \bibinfo
  {author} {\bibfnamefont {R.}~\bibnamefont {Geiss}}, \ and\ \bibinfo {author}
  {\bibfnamefont {T.}~\bibnamefont {Beth}},\ }\href@noop {} {\bibfield
  {journal} {\bibinfo  {journal} {Int. J. Theor. Phys.}\ }\textbf {\bibinfo
  {volume} {39}},\ \bibinfo {pages} {2717} (\bibinfo {year}
  {2000})}\BibitemShut {NoStop}%
\bibitem [{\citenamefont {Horodecki}\ and\ \citenamefont
  {Oppenheim}(2013)}]{FundLimits2}%
  \BibitemOpen
  \bibfield  {author} {\bibinfo {author} {\bibfnamefont {M.}~\bibnamefont
  {Horodecki}}\ and\ \bibinfo {author} {\bibfnamefont {J.}~\bibnamefont
  {Oppenheim}},\ }\href@noop {} {\bibfield  {journal} {\bibinfo  {journal}
  {Nat. Commun.}\ }\textbf {\bibinfo {volume} {4}},\ \bibinfo {pages} {1}
  (\bibinfo {year} {2013})}\BibitemShut {NoStop}%
\bibitem [{\citenamefont {Brand\~ao}\ \emph {et~al.}(2013)\citenamefont
  {Brand\~ao}, \citenamefont {Horodecki}, \citenamefont {Oppenheim},
  \citenamefont {Renes},\ and\ \citenamefont {Spekkens}}]{BrandaoHORS13}%
  \BibitemOpen
  \bibfield  {author} {\bibinfo {author} {\bibfnamefont {F.~G. S.~L.}\
  \bibnamefont {Brand\~ao}}, \bibinfo {author} {\bibfnamefont {M.}~\bibnamefont
  {Horodecki}}, \bibinfo {author} {\bibfnamefont {J.}~\bibnamefont
  {Oppenheim}}, \bibinfo {author} {\bibfnamefont {J.~M.}\ \bibnamefont
  {Renes}}, \ and\ \bibinfo {author} {\bibfnamefont {R.~W.}\ \bibnamefont
  {Spekkens}},\ }\href {\doibase 10.1103/PhysRevLett.111.250404} {\bibfield
  {journal} {\bibinfo  {journal} {Physical Review Letters}\ }\textbf {\bibinfo
  {volume} {111}},\ \bibinfo {pages} {250404} (\bibinfo {year}
  {2013})}\BibitemShut {NoStop}%
\bibitem [{\citenamefont {Brand{\~a}o}\ \emph {et~al.}(2015)\citenamefont
  {Brand{\~a}o}, \citenamefont {Horodecki}, \citenamefont {Ng}, \citenamefont
  {Oppenheim},\ and\ \citenamefont {Wehner}}]{BrandaoHNOW14}%
  \BibitemOpen
  \bibfield  {author} {\bibinfo {author} {\bibfnamefont {F.~G.}\ \bibnamefont
  {Brand{\~a}o}}, \bibinfo {author} {\bibfnamefont {M.}~\bibnamefont
  {Horodecki}}, \bibinfo {author} {\bibfnamefont {N.~H.~Y.}\ \bibnamefont
  {Ng}}, \bibinfo {author} {\bibfnamefont {J.}~\bibnamefont {Oppenheim}}, \
  and\ \bibinfo {author} {\bibfnamefont {S.}~\bibnamefont {Wehner}},\ }\href
  {\doibase 10.1073/pnas.1411728112} {\bibfield  {journal} {\bibinfo  {journal}
  {Proc. Nat. Acad. Sci. USA}\ }\textbf {\bibinfo {volume} {112}},\ \bibinfo
  {pages} {3275} (\bibinfo {year} {2015})}\BibitemShut {NoStop}%
\bibitem [{\citenamefont {Horodecki}\ \emph
  {et~al.}(2003{\natexlab{a}})\citenamefont {Horodecki}, \citenamefont
  {Horodecki},\ and\ \citenamefont {Oppenheim}}]{HHOLong}%
  \BibitemOpen
  \bibfield  {author} {\bibinfo {author} {\bibfnamefont {M.}~\bibnamefont
  {Horodecki}}, \bibinfo {author} {\bibfnamefont {P.}~\bibnamefont
  {Horodecki}}, \ and\ \bibinfo {author} {\bibfnamefont {J.}~\bibnamefont
  {Oppenheim}},\ }\href {\doibase 10.1103/PhysRevA.67.062104} {\bibfield
  {journal} {\bibinfo  {journal} {Phys. Rev. A}\ }\textbf {\bibinfo {volume}
  {67}},\ \bibinfo {pages} {062104} (\bibinfo {year}
  {2003}{\natexlab{a}})}\BibitemShut {NoStop}%
\bibitem [{\citenamefont {Gour}\ \emph {et~al.}(2015)\citenamefont {Gour},
  \citenamefont {Müller}, \citenamefont {Narasimhachar}, \citenamefont
  {Spekkens},\ and\ \citenamefont {Halpern}}]{GourMNSYH13}%
  \BibitemOpen
  \bibfield  {author} {\bibinfo {author} {\bibfnamefont {G.}~\bibnamefont
  {Gour}}, \bibinfo {author} {\bibfnamefont {M.~P.}\ \bibnamefont {Müller}},
  \bibinfo {author} {\bibfnamefont {V.}~\bibnamefont {Narasimhachar}}, \bibinfo
  {author} {\bibfnamefont {R.~W.}\ \bibnamefont {Spekkens}}, \ and\ \bibinfo
  {author} {\bibfnamefont {N.~Y.}\ \bibnamefont {Halpern}},\ }\href {\doibase
  https://doi.org/10.1016/j.physrep.2015.04.003} {\bibfield  {journal}
  {\bibinfo  {journal} {Physics Reports}\ }\textbf {\bibinfo {volume} {583}},\
  \bibinfo {pages} {1 } (\bibinfo {year} {2015})},\ \bibinfo {note} {the
  resource theory of informational nonequilibrium in
  thermodynamics}\BibitemShut {NoStop}%
\bibitem [{\citenamefont {Faist}\ \emph
  {et~al.}(2015{\natexlab{a}})\citenamefont {Faist}, \citenamefont {Dupuis},
  \citenamefont {Oppenheim},\ and\ \citenamefont {Renner}}]{FaistDOR12}%
  \BibitemOpen
  \bibfield  {author} {\bibinfo {author} {\bibfnamefont {P.}~\bibnamefont
  {Faist}}, \bibinfo {author} {\bibfnamefont {F.}~\bibnamefont {Dupuis}},
  \bibinfo {author} {\bibfnamefont {J.}~\bibnamefont {Oppenheim}}, \ and\
  \bibinfo {author} {\bibfnamefont {R.}~\bibnamefont {Renner}},\ }\href
  {http://dx.doi.org/10.1038/ncomms8669} {\bibfield  {journal} {\bibinfo
  {journal} {Nature Communications}\ }\textbf {\bibinfo {volume} {6}},\
  \bibinfo {pages} {7669 EP } (\bibinfo {year} {2015}{\natexlab{a}})},\
  \bibinfo {note} {article}\BibitemShut {NoStop}%
\bibitem [{\citenamefont {Yunger~Halpern}\ and\ \citenamefont
  {Renes}(2016)}]{YungerHalpernR14}%
  \BibitemOpen
  \bibfield  {author} {\bibinfo {author} {\bibfnamefont {N.}~\bibnamefont
  {Yunger~Halpern}}\ and\ \bibinfo {author} {\bibfnamefont {J.~M.}\
  \bibnamefont {Renes}},\ }\href {\doibase 10.1103/PhysRevE.93.022126}
  {\bibfield  {journal} {\bibinfo  {journal} {Phys. Rev. E}\ }\textbf {\bibinfo
  {volume} {93}},\ \bibinfo {pages} {022126} (\bibinfo {year}
  {2016})}\BibitemShut {NoStop}%
\bibitem [{\citenamefont {Dahlsten}(2013)}]{Dahlsten13}%
  \BibitemOpen
  \bibfield  {author} {\bibinfo {author} {\bibfnamefont {O.~C.~O.}\
  \bibnamefont {Dahlsten}},\ }\href {\doibase 10.3390/e15125346} {\bibfield
  {journal} {\bibinfo  {journal} {Entropy}\ }\textbf {\bibinfo {volume} {15}},\
  \bibinfo {pages} {5346} (\bibinfo {year} {2013})}\BibitemShut {NoStop}%
\bibitem [{\citenamefont {{Renner}}(2005)}]{RennerThesis}%
  \BibitemOpen
  \bibfield  {author} {\bibinfo {author} {\bibfnamefont {R.}~\bibnamefont
  {{Renner}}},\ }\emph {\bibinfo {title} {{Security of Quantum Key
  Distribution}}},\ \href@noop {} {Ph.D. thesis},\ \bibinfo  {school} {PhD
  Thesis, 2005} (\bibinfo {year} {2005})\BibitemShut {NoStop}%
\bibitem [{\citenamefont {{Horodecki}}\ and\ \citenamefont
  {{Oppenheim}}(2013)}]{HorodeckiO13}%
  \BibitemOpen
  \bibfield  {author} {\bibinfo {author} {\bibfnamefont {M.}~\bibnamefont
  {{Horodecki}}}\ and\ \bibinfo {author} {\bibfnamefont {J.}~\bibnamefont
  {{Oppenheim}}},\ }\href {\doibase 10.1142/S0217979213450197} {\bibfield
  {journal} {\bibinfo  {journal} {International Journal of Modern Physics B}\
  }\textbf {\bibinfo {volume} {27}},\ \bibinfo {eid} {1345019} (\bibinfo {year}
  {2013})}\BibitemShut {NoStop}%
\bibitem [{\citenamefont {Coecke}\ \emph {et~al.}(2016)\citenamefont {Coecke},
  \citenamefont {Fritz},\ and\ \citenamefont {Spekkens}}]{CoeckeFS14}%
  \BibitemOpen
  \bibfield  {author} {\bibinfo {author} {\bibfnamefont {B.}~\bibnamefont
  {Coecke}}, \bibinfo {author} {\bibfnamefont {T.}~\bibnamefont {Fritz}}, \
  and\ \bibinfo {author} {\bibfnamefont {R.~W.}\ \bibnamefont {Spekkens}},\
  }\href {\doibase https://doi.org/10.1016/j.ic.2016.02.008} {\bibfield
  {journal} {\bibinfo  {journal} {Information and Computation}\ }\textbf
  {\bibinfo {volume} {250}},\ \bibinfo {pages} {59 } (\bibinfo {year}
  {2016})},\ \bibinfo {note} {quantum Physics and Logic}\BibitemShut {NoStop}%
\bibitem [{\citenamefont {Horodecki}\ \emph {et~al.}(2009)\citenamefont
  {Horodecki}, \citenamefont {Horodecki}, \citenamefont {Horodecki},\ and\
  \citenamefont {Horodecki}}]{HorodeckiHHH09}%
  \BibitemOpen
  \bibfield  {author} {\bibinfo {author} {\bibfnamefont {R.}~\bibnamefont
  {Horodecki}}, \bibinfo {author} {\bibfnamefont {P.}~\bibnamefont
  {Horodecki}}, \bibinfo {author} {\bibfnamefont {M.}~\bibnamefont
  {Horodecki}}, \ and\ \bibinfo {author} {\bibfnamefont {K.}~\bibnamefont
  {Horodecki}},\ }\href@noop {} {\bibfield  {journal} {\bibinfo  {journal}
  {Rev. Mod. Phys.}\ }\textbf {\bibinfo {volume} {81}},\ \bibinfo {pages} {865}
  (\bibinfo {year} {2009})}\BibitemShut {NoStop}%
\bibitem [{\citenamefont {Bartlett}\ \emph {et~al.}(2007)\citenamefont
  {Bartlett}, \citenamefont {Rudolph},\ and\ \citenamefont
  {Spekkens}}]{BartlettRS07}%
  \BibitemOpen
  \bibfield  {author} {\bibinfo {author} {\bibfnamefont {S.~D.}\ \bibnamefont
  {Bartlett}}, \bibinfo {author} {\bibfnamefont {T.}~\bibnamefont {Rudolph}}, \
  and\ \bibinfo {author} {\bibfnamefont {R.~W.}\ \bibnamefont {Spekkens}},\
  }\href {\doibase 10.1103/RevModPhys.79.555} {\bibfield  {journal} {\bibinfo
  {journal} {Rev. Mod. Phys.}\ }\textbf {\bibinfo {volume} {79}},\ \bibinfo
  {pages} {555} (\bibinfo {year} {2007})}\BibitemShut {NoStop}%
\bibitem [{\citenamefont {Marvian}\ and\ \citenamefont
  {Spekkens}(2013)}]{marvian_theory_2013}%
  \BibitemOpen
  \bibfield  {author} {\bibinfo {author} {\bibfnamefont {I.}~\bibnamefont
  {Marvian}}\ and\ \bibinfo {author} {\bibfnamefont {R.~W.}\ \bibnamefont
  {Spekkens}},\ }\href {\doibase 10.1088/1367-2630/15/3/033001} {\bibfield
  {journal} {\bibinfo  {journal} {New Journal of Physics}\ }\textbf {\bibinfo
  {volume} {15}},\ \bibinfo {pages} {033001} (\bibinfo {year}
  {2013})}\BibitemShut {NoStop}%
\bibitem [{\citenamefont {Bartlett}\ \emph {et~al.}(2006)\citenamefont
  {Bartlett}, \citenamefont {Rudolph}, \citenamefont {Spekkens},\ and\
  \citenamefont {Turner}}]{BartlettRST06}%
  \BibitemOpen
  \bibfield  {author} {\bibinfo {author} {\bibfnamefont {S.~D.}\ \bibnamefont
  {Bartlett}}, \bibinfo {author} {\bibfnamefont {T.}~\bibnamefont {Rudolph}},
  \bibinfo {author} {\bibfnamefont {R.~W.}\ \bibnamefont {Spekkens}}, \ and\
  \bibinfo {author} {\bibfnamefont {P.~S.}\ \bibnamefont {Turner}},\
  }\href@noop {} {\bibfield  {journal} {\bibinfo  {journal} {New J. Phys.}\
  }\textbf {\bibinfo {volume} {8}},\ \bibinfo {pages} {58} (\bibinfo {year}
  {2006})}\BibitemShut {NoStop}%
\bibitem [{\citenamefont {{Veitch}}\ \emph {et~al.}(2014)\citenamefont
  {{Veitch}}, \citenamefont {{Hamed Mousavian}}, \citenamefont {{Gottesman}},\
  and\ \citenamefont {{Emerson}}}]{VeitchMGE14}%
  \BibitemOpen
  \bibfield  {author} {\bibinfo {author} {\bibfnamefont {V.}~\bibnamefont
  {{Veitch}}}, \bibinfo {author} {\bibfnamefont {S.~A.}\ \bibnamefont {{Hamed
  Mousavian}}}, \bibinfo {author} {\bibfnamefont {D.}~\bibnamefont
  {{Gottesman}}}, \ and\ \bibinfo {author} {\bibfnamefont {J.}~\bibnamefont
  {{Emerson}}},\ }\href {\doibase 10.1088/1367-2630/16/1/013009} {\bibfield
  {journal} {\bibinfo  {journal} {New Journal of Physics}\ }\textbf {\bibinfo
  {volume} {16}},\ \bibinfo {eid} {013009} (\bibinfo {year}
  {2014})}\BibitemShut {NoStop}%
\bibitem [{\citenamefont {Skrzypczyk}\ \emph {et~al.}(2013)\citenamefont
  {Skrzypczyk}, \citenamefont {Short},\ and\ \citenamefont
  {Popescu}}]{SkrzypczykSP13Extract}%
  \BibitemOpen
  \bibfield  {author} {\bibinfo {author} {\bibfnamefont {P.}~\bibnamefont
  {Skrzypczyk}}, \bibinfo {author} {\bibfnamefont {A.~J.}\ \bibnamefont
  {Short}}, \ and\ \bibinfo {author} {\bibfnamefont {S.}~\bibnamefont
  {Popescu}},\ }\href@noop {} {\bibfield  {journal} {\bibinfo  {journal}
  {arXiv:1302.2811}\ } (\bibinfo {year} {2013})}\BibitemShut {NoStop}%
\bibitem [{\citenamefont {Horodecki}\ \emph
  {et~al.}(2003{\natexlab{b}})\citenamefont {Horodecki}, \citenamefont
  {Horodecki}, \citenamefont {Horodecki}, \citenamefont {Horodecki},
  \citenamefont {Oppenheim}, \citenamefont {Sen(De)},\ and\ \citenamefont
  {Sen}}]{HHOShort}%
  \BibitemOpen
  \bibfield  {author} {\bibinfo {author} {\bibfnamefont {M.}~\bibnamefont
  {Horodecki}}, \bibinfo {author} {\bibfnamefont {K.}~\bibnamefont
  {Horodecki}}, \bibinfo {author} {\bibfnamefont {P.}~\bibnamefont
  {Horodecki}}, \bibinfo {author} {\bibfnamefont {R.}~\bibnamefont
  {Horodecki}}, \bibinfo {author} {\bibfnamefont {J.}~\bibnamefont
  {Oppenheim}}, \bibinfo {author} {\bibfnamefont {A.}~\bibnamefont {Sen(De)}},
  \ and\ \bibinfo {author} {\bibfnamefont {U.}~\bibnamefont {Sen}},\
  }\href@noop {} {\bibfield  {journal} {\bibinfo  {journal} {Phys. Rev. Lett.}\
  }\textbf {\bibinfo {volume} {90}},\ \bibinfo {pages} {100402} (\bibinfo
  {year} {2003}{\natexlab{b}})}\BibitemShut {NoStop}%
\bibitem [{\citenamefont {Callen}(1985)}]{Callen85}%
  \BibitemOpen
  \bibfield  {author} {\bibinfo {author} {\bibfnamefont {H.~B.}\ \bibnamefont
  {Callen}},\ }\href@noop {} {\emph {\bibinfo {title} {Thermodynamics and an
  Introduction to Thermostatistics}}},\ \bibinfo {edition} {2nd}\ ed.\
  (\bibinfo  {publisher} {John Wiley \& Sons},\ \bibinfo {address} {New York},\
  \bibinfo {year} {1985})\BibitemShut {NoStop}%
\bibitem [{\citenamefont {Jaynes}(1957{\natexlab{a}})}]{JaynesII57}%
  \BibitemOpen
  \bibfield  {author} {\bibinfo {author} {\bibfnamefont {E.~T.}\ \bibnamefont
  {Jaynes}},\ }\href@noop {} {\bibfield  {journal} {\bibinfo  {journal} {Phys.
  Rev.}\ }\textbf {\bibinfo {volume} {108}},\ \bibinfo {pages} {171} (\bibinfo
  {year} {1957}{\natexlab{a}})}\BibitemShut {NoStop}%
\bibitem [{\citenamefont {Lostaglio}\ \emph {et~al.}(2017)\citenamefont
  {Lostaglio}, \citenamefont {Jennings},\ and\ \citenamefont
  {Rudolph}}]{Lostaglio_17_Thermodynamic}%
  \BibitemOpen
  \bibfield  {author} {\bibinfo {author} {\bibfnamefont {M.}~\bibnamefont
  {Lostaglio}}, \bibinfo {author} {\bibfnamefont {D.}~\bibnamefont {Jennings}},
  \ and\ \bibinfo {author} {\bibfnamefont {T.}~\bibnamefont {Rudolph}},\ }\href
  {http://stacks.iop.org/1367-2630/19/i=4/a=043008} {\bibfield  {journal}
  {\bibinfo  {journal} {New Journal of Physics}\ }\textbf {\bibinfo {volume}
  {19}},\ \bibinfo {pages} {043008} (\bibinfo {year} {2017})}\BibitemShut
  {NoStop}%
\bibitem [{\citenamefont {Guryanova}\ \emph {et~al.}(2016)\citenamefont
  {Guryanova}, \citenamefont {Popescu}, \citenamefont {Short}, \citenamefont
  {Silva},\ and\ \citenamefont {Skrzypczyk}}]{Guryanova_16_Thermodynamics}%
  \BibitemOpen
  \bibfield  {author} {\bibinfo {author} {\bibfnamefont {Y.}~\bibnamefont
  {Guryanova}}, \bibinfo {author} {\bibfnamefont {S.}~\bibnamefont {Popescu}},
  \bibinfo {author} {\bibfnamefont {A.~J.}\ \bibnamefont {Short}}, \bibinfo
  {author} {\bibfnamefont {R.}~\bibnamefont {Silva}}, \ and\ \bibinfo {author}
  {\bibfnamefont {P.}~\bibnamefont {Skrzypczyk}},\ }\href
  {http://dx.doi.org/10.1038/ncomms12049} {\bibfield  {journal} {\bibinfo
  {journal} {Nature Communications}\ }\textbf {\bibinfo {volume} {7}},\
  \bibinfo {pages} {12049 EP } (\bibinfo {year} {2016})},\ \bibinfo {note}
  {article}\BibitemShut {NoStop}%
\bibitem [{\citenamefont {Yunger~Halpern}\ \emph {et~al.}(2016)\citenamefont
  {Yunger~Halpern}, \citenamefont {Faist}, \citenamefont {Oppenheim},\ and\
  \citenamefont {Winter}}]{YH_16_Microcanonical}%
  \BibitemOpen
  \bibfield  {author} {\bibinfo {author} {\bibfnamefont {N.}~\bibnamefont
  {Yunger~Halpern}}, \bibinfo {author} {\bibfnamefont {P.}~\bibnamefont
  {Faist}}, \bibinfo {author} {\bibfnamefont {J.}~\bibnamefont {Oppenheim}}, \
  and\ \bibinfo {author} {\bibfnamefont {A.}~\bibnamefont {Winter}},\ }\href
  {http://dx.doi.org/10.1038/ncomms12051} {\bibfield  {journal} {\bibinfo
  {journal} {Nature Communications}\ }\textbf {\bibinfo {volume} {7}},\
  \bibinfo {pages} {12051 EP } (\bibinfo {year} {2016})},\ \bibinfo {note}
  {article}\BibitemShut {NoStop}%
\bibitem [{\citenamefont {{Ito}}\ and\ \citenamefont
  {{Hayashi}}(2016)}]{Ito_16_Optimal}%
  \BibitemOpen
  \bibfield  {author} {\bibinfo {author} {\bibfnamefont {K.}~\bibnamefont
  {{Ito}}}\ and\ \bibinfo {author} {\bibfnamefont {M.}~\bibnamefont
  {{Hayashi}}},\ }\href@noop {} {\bibfield  {journal} {\bibinfo  {journal}
  {ArXiv e-prints}\ } (\bibinfo {year} {2016})},\ \Eprint
  {http://arxiv.org/abs/1612.04047} {arXiv:1612.04047 [quant-ph]} \BibitemShut
  {NoStop}%
\bibitem [{\citenamefont {{Nath Bera}}\ \emph {et~al.}(2017)\citenamefont
  {{Nath Bera}}, \citenamefont {{Riera}}, \citenamefont {{Lewenstein}},\ and\
  \citenamefont {{Winter}}}]{NathBera_17_Thermodynamics}%
  \BibitemOpen
  \bibfield  {author} {\bibinfo {author} {\bibfnamefont {M.}~\bibnamefont
  {{Nath Bera}}}, \bibinfo {author} {\bibfnamefont {A.}~\bibnamefont
  {{Riera}}}, \bibinfo {author} {\bibfnamefont {M.}~\bibnamefont
  {{Lewenstein}}}, \ and\ \bibinfo {author} {\bibfnamefont {A.}~\bibnamefont
  {{Winter}}},\ }\href@noop {} {\bibfield  {journal} {\bibinfo  {journal}
  {ArXiv e-prints}\ } (\bibinfo {year} {2017})},\ \Eprint
  {http://arxiv.org/abs/1707.01750} {arXiv:1707.01750 [quant-ph]} \BibitemShut
  {NoStop}%
\bibitem [{\citenamefont {{Mur-Petit}}\ \emph {et~al.}(2017)\citenamefont
  {{Mur-Petit}}, \citenamefont {{Rela{\~n}o}}, \citenamefont {{Molina}},\ and\
  \citenamefont {{Jaksch}}}]{MurPetit_17_Generalised}%
  \BibitemOpen
  \bibfield  {author} {\bibinfo {author} {\bibfnamefont {J.}~\bibnamefont
  {{Mur-Petit}}}, \bibinfo {author} {\bibfnamefont {A.}~\bibnamefont
  {{Rela{\~n}o}}}, \bibinfo {author} {\bibfnamefont {R.~A.}\ \bibnamefont
  {{Molina}}}, \ and\ \bibinfo {author} {\bibfnamefont {D.}~\bibnamefont
  {{Jaksch}}},\ }\href@noop {} {\bibfield  {journal} {\bibinfo  {journal}
  {ArXiv e-prints}\ } (\bibinfo {year} {2017})},\ \Eprint
  {http://arxiv.org/abs/1711.00871} {arXiv:1711.00871 [quant-ph]} \BibitemShut
  {NoStop}%
\bibitem [{\citenamefont {Lostaglio}(2014)}]{Lostaglio_14_Masters}%
  \BibitemOpen
  \bibfield  {author} {\bibinfo {author} {\bibfnamefont {M.}~\bibnamefont
  {Lostaglio}},\ }\emph {\bibinfo {title} {The resource theory of quantum
  thermodynamics}},\ \href@noop {} {Master's thesis},\ \bibinfo  {school}
  {Imperial College London} (\bibinfo {year} {2014})\BibitemShut {NoStop}%
\bibitem [{\citenamefont {Vaccaro}\ and\ \citenamefont
  {Barnett}(2011)}]{VaccaroB11}%
  \BibitemOpen
  \bibfield  {author} {\bibinfo {author} {\bibfnamefont {J.~A.}\ \bibnamefont
  {Vaccaro}}\ and\ \bibinfo {author} {\bibfnamefont {S.~M.}\ \bibnamefont
  {Barnett}},\ }\href {\doibase 10.1098/rspa.2010.0577} {\bibfield  {journal}
  {\bibinfo  {journal} {Proceedings of the Royal Society of London A:
  Mathematical, Physical and Engineering Sciences}\ }\textbf {\bibinfo {volume}
  {467}},\ \bibinfo {pages} {1770} (\bibinfo {year} {2011})}\BibitemShut
  {NoStop}%
\bibitem [{\citenamefont {Barnett}\ and\ \citenamefont
  {Vaccaro}(2013)}]{BarnettV13}%
  \BibitemOpen
  \bibfield  {author} {\bibinfo {author} {\bibfnamefont {S.~M.}\ \bibnamefont
  {Barnett}}\ and\ \bibinfo {author} {\bibfnamefont {J.~A.}\ \bibnamefont
  {Vaccaro}},\ }\href {\doibase 10.3390/e15114956} {\bibfield  {journal}
  {\bibinfo  {journal} {Entropy}\ }\textbf {\bibinfo {volume} {15}},\ \bibinfo
  {pages} {4956} (\bibinfo {year} {2013})}\BibitemShut {NoStop}%
\bibitem [{\citenamefont {Alberty}(2001)}]{Alberty01}%
  \BibitemOpen
  \bibfield  {author} {\bibinfo {author} {\bibfnamefont {R.~A.}\ \bibnamefont
  {Alberty}},\ }\href@noop {} {\bibfield  {journal} {\bibinfo  {journal} {Pure
  and Applied Chemistry}\ }\textbf {\bibinfo {volume} {73}},\ \bibinfo {pages}
  {1349Ð1380} (\bibinfo {year} {2001})}\BibitemShut {NoStop}%
\bibitem [{\citenamefont {Quan}(2009)}]{Quan_09_Quantum}%
  \BibitemOpen
  \bibfield  {author} {\bibinfo {author} {\bibfnamefont {H.~T.}\ \bibnamefont
  {Quan}},\ }\href {\doibase 10.1103/PhysRevE.79.041129} {\bibfield  {journal}
  {\bibinfo  {journal} {Phys. Rev. E}\ }\textbf {\bibinfo {volume} {79}},\
  \bibinfo {pages} {041129} (\bibinfo {year} {2009})}\BibitemShut {NoStop}%
\bibitem [{\citenamefont {Reichl}(1980)}]{Reichl80}%
  \BibitemOpen
  \bibfield  {author} {\bibinfo {author} {\bibfnamefont {L.~E.}\ \bibnamefont
  {Reichl}},\ }\href@noop {} {\emph {\bibinfo {title} {A Modern Course in
  Statistical Physics}}}\ (\bibinfo  {publisher} {U. of Texas Press},\ \bibinfo
  {address} {Austin, Texas},\ \bibinfo {year} {1980})\BibitemShut {NoStop}%
\bibitem [{\citenamefont {Belof}(2009)}]{Belof09}%
  \BibitemOpen
  \bibfield  {author} {\bibinfo {author} {\bibfnamefont {J.~L.}\ \bibnamefont
  {Belof}},\ }\emph {\bibinfo {title} {Theory and simulation of metal-organic
  materials and biomolecules}},\ \href@noop {} {Ph.D. thesis},\ \bibinfo
  {school} {U. of South Florida} (\bibinfo {year} {2009})\BibitemShut {NoStop}%
\bibitem [{\citenamefont {Lenard}(1978)}]{Lenard78}%
  \BibitemOpen
  \bibfield  {author} {\bibinfo {author} {\bibfnamefont {A.}~\bibnamefont
  {Lenard}},\ }\href {\doibase 10.1007/BF01011769} {\bibfield  {journal}
  {\bibinfo  {journal} {Journal of Statistical Physics}\ }\textbf {\bibinfo
  {volume} {19}},\ \bibinfo {pages} {575} (\bibinfo {year} {1978})}\BibitemShut
  {NoStop}%
\bibitem [{\citenamefont {Marshall}\ \emph {et~al.}(2010)\citenamefont
  {Marshall}, \citenamefont {Olkin},\ and\ \citenamefont
  {Arnold}}]{MarshallOA10}%
  \BibitemOpen
  \bibfield  {author} {\bibinfo {author} {\bibfnamefont {A.~W.}\ \bibnamefont
  {Marshall}}, \bibinfo {author} {\bibfnamefont {I.}~\bibnamefont {Olkin}}, \
  and\ \bibinfo {author} {\bibfnamefont {B.~C.}\ \bibnamefont {Arnold}},\
  }\href@noop {} {\emph {\bibinfo {title} {{Inequalities: theory of
  majorization and its applications}}}}\ (\bibinfo  {publisher} {Springer},\
  \bibinfo {year} {2010})\BibitemShut {NoStop}%
\bibitem [{\citenamefont {Ruch}\ \emph {et~al.}(1978)\citenamefont {Ruch},
  \citenamefont {Schranner},\ and\ \citenamefont {Seligman}}]{RuchSS78}%
  \BibitemOpen
  \bibfield  {author} {\bibinfo {author} {\bibfnamefont {E.}~\bibnamefont
  {Ruch}}, \bibinfo {author} {\bibfnamefont {R.}~\bibnamefont {Schranner}}, \
  and\ \bibinfo {author} {\bibfnamefont {T.~H.}\ \bibnamefont {Seligman}},\
  }\href@noop {} {\bibfield  {journal} {\bibinfo  {journal} {The Journal of
  Chemical Physics}\ }\textbf {\bibinfo {volume} {69}},\ \bibinfo {pages} {386}
  (\bibinfo {year} {1978})}\BibitemShut {NoStop}%
\bibitem [{\citenamefont {Talkner}\ \emph {et~al.}(2007)\citenamefont
  {Talkner}, \citenamefont {Lutz},\ and\ \citenamefont
  {H\"{a}nggi}}]{TalknerLH07}%
  \BibitemOpen
  \bibfield  {author} {\bibinfo {author} {\bibfnamefont {P.}~\bibnamefont
  {Talkner}}, \bibinfo {author} {\bibfnamefont {E.}~\bibnamefont {Lutz}}, \
  and\ \bibinfo {author} {\bibfnamefont {P.}~\bibnamefont {H\"{a}nggi}},\
  }\href {\doibase 10.1103/PhysRevE.75.050102} {\bibfield  {journal} {\bibinfo
  {journal} {Physical Review E}\ }\textbf {\bibinfo {volume} {75}},\ \bibinfo
  {pages} {050102} (\bibinfo {year} {2007})}\BibitemShut {NoStop}%
\bibitem [{\citenamefont {{Frenzel}}\ \emph {et~al.}(2014)\citenamefont
  {{Frenzel}}, \citenamefont {{Jennings}},\ and\ \citenamefont
  {{Rudolph}}}]{FrenzelDR14}%
  \BibitemOpen
  \bibfield  {author} {\bibinfo {author} {\bibfnamefont {M.~F.}\ \bibnamefont
  {{Frenzel}}}, \bibinfo {author} {\bibfnamefont {D.}~\bibnamefont
  {{Jennings}}}, \ and\ \bibinfo {author} {\bibfnamefont {T.}~\bibnamefont
  {{Rudolph}}},\ }\href {\doibase 10.1103/PhysRevE.90.052136} {\bibfield
  {journal} {\bibinfo  {journal} {Phys. Rev. E}\ }\textbf {\bibinfo {volume}
  {90}},\ \bibinfo {eid} {052136} (\bibinfo {year} {2014})}\BibitemShut
  {NoStop}%
\bibitem [{\citenamefont {Jaynes}(1957{\natexlab{b}})}]{JaynesI57}%
  \BibitemOpen
  \bibfield  {author} {\bibinfo {author} {\bibfnamefont {E.~T.}\ \bibnamefont
  {Jaynes}},\ }\href@noop {} {\bibfield  {journal} {\bibinfo  {journal} {Phys.
  Rev.}\ }\textbf {\bibinfo {volume} {106}},\ \bibinfo {pages} {620} (\bibinfo
  {year} {1957}{\natexlab{b}})}\BibitemShut {NoStop}%
\bibitem [{\citenamefont {Guggenheim}(1959)}]{Guggenheim59}%
  \BibitemOpen
  \bibfield  {author} {\bibinfo {author} {\bibfnamefont {E.~A.}\ \bibnamefont
  {Guggenheim}},\ }\href@noop {} {\emph {\bibinfo {title} {Thermodynamics: An
  Advanced Treatment for Chemists and Physicists}}}\ (\bibinfo  {publisher}
  {North-Holland},\ \bibinfo {year} {1959})\BibitemShut {NoStop}%
\bibitem [{\citenamefont {Kirkwood}\ and\ \citenamefont
  {Oppenheim}(1961)}]{KirkwoodO61}%
  \BibitemOpen
  \bibfield  {author} {\bibinfo {author} {\bibfnamefont {J.~G.}\ \bibnamefont
  {Kirkwood}}\ and\ \bibinfo {author} {\bibfnamefont {I.}~\bibnamefont
  {Oppenheim}},\ }\href@noop {} {\emph {\bibinfo {title} {Chemical
  Thermodynamics}}}\ (\bibinfo  {publisher} {McGraw-Hill},\ \bibinfo {year}
  {1961})\BibitemShut {NoStop}%
\bibitem [{\citenamefont {Szilard}(1929)}]{Szilard29}%
  \BibitemOpen
  \bibfield  {author} {\bibinfo {author} {\bibfnamefont {L.}~\bibnamefont
  {Szilard}},\ }\href@noop {} {\bibfield  {journal} {\bibinfo  {journal}
  {Zeitschrift f{\"u}r Physik}\ }\textbf {\bibinfo {volume} {53}},\ \bibinfo
  {pages} {840} (\bibinfo {year} {1929})}\BibitemShut {NoStop}%
\bibitem [{\citenamefont {Hiai}\ and\ \citenamefont
  {Petz}(1991)}]{Hiai_91_Proper}%
  \BibitemOpen
  \bibfield  {author} {\bibinfo {author} {\bibfnamefont {F.}~\bibnamefont
  {Hiai}}\ and\ \bibinfo {author} {\bibfnamefont {D.}~\bibnamefont {Petz}},\
  }\href {\doibase 10.1007/BF02100287} {\bibfield  {journal} {\bibinfo
  {journal} {Communications in Mathematical Physics}\ }\textbf {\bibinfo
  {volume} {143}},\ \bibinfo {pages} {99} (\bibinfo {year} {1991})}\BibitemShut
  {NoStop}%
\bibitem [{\citenamefont {Dupuis}\ \emph {et~al.}(2013)\citenamefont {Dupuis},
  \citenamefont {Kramer}, \citenamefont {Faist}, \citenamefont {Renes},\ and\
  \citenamefont {Renner}}]{DupuisKFRR12}%
  \BibitemOpen
  \bibfield  {author} {\bibinfo {author} {\bibfnamefont {F.}~\bibnamefont
  {Dupuis}}, \bibinfo {author} {\bibfnamefont {L.}~\bibnamefont {Kramer}},
  \bibinfo {author} {\bibfnamefont {P.}~\bibnamefont {Faist}}, \bibinfo
  {author} {\bibfnamefont {J.~M.}\ \bibnamefont {Renes}}, \ and\ \bibinfo
  {author} {\bibfnamefont {R.}~\bibnamefont {Renner}},\ }\enquote {\bibinfo
  {title} {Generalized entropies},}\ in\ \href {\doibase
  10.1142/9789814449243_0008} {\emph {\bibinfo {booktitle} {XVIIth
  International Congress on Mathematical Physics}}}\ (\bibinfo  {publisher}
  {World Scientific},\ \bibinfo {year} {2013})\ pp.\ \bibinfo {pages}
  {134--153}\BibitemShut {NoStop}%
\bibitem [{\citenamefont {{Linden}}\ \emph {et~al.}(1999)\citenamefont
  {{Linden}}, \citenamefont {{Popescu}}, \citenamefont {{Schumacher}},\ and\
  \citenamefont {{Westmoreland}}}]{LindenPSW99}%
  \BibitemOpen
  \bibfield  {author} {\bibinfo {author} {\bibfnamefont {N.}~\bibnamefont
  {{Linden}}}, \bibinfo {author} {\bibfnamefont {S.}~\bibnamefont {{Popescu}}},
  \bibinfo {author} {\bibfnamefont {B.}~\bibnamefont {{Schumacher}}}, \ and\
  \bibinfo {author} {\bibfnamefont {M.}~\bibnamefont {{Westmoreland}}},\
  }\href@noop {} {\bibfield  {journal} {\bibinfo  {journal}
  {arXiv:quant-ph/9912039}\ } (\bibinfo {year} {1999})},\ \Eprint
  {http://arxiv.org/abs/quant-ph/9912039} {quant-ph/9912039} \BibitemShut
  {NoStop}%
\bibitem [{\citenamefont {{Vidal}}\ \emph {et~al.}(2000)\citenamefont
  {{Vidal}}, \citenamefont {{D{\"u}r}},\ and\ \citenamefont {{Cirac}}}]{VDC00}%
  \BibitemOpen
  \bibfield  {author} {\bibinfo {author} {\bibfnamefont {G.}~\bibnamefont
  {{Vidal}}}, \bibinfo {author} {\bibfnamefont {W.}~\bibnamefont {{D{\"u}r}}},
  \ and\ \bibinfo {author} {\bibfnamefont {J.~I.}\ \bibnamefont {{Cirac}}},\
  }\href {\doibase 10.1103/PhysRevLett.85.658} {\bibfield  {journal} {\bibinfo
  {journal} {Physical Review Letters}\ }\textbf {\bibinfo {volume} {85}},\
  \bibinfo {pages} {658} (\bibinfo {year} {2000})},\ \Eprint
  {http://arxiv.org/abs/quant-ph/0004009} {quant-ph/0004009} \BibitemShut
  {NoStop}%
\bibitem [{\citenamefont {{Wu}}\ and\ \citenamefont {{Zhang}}(2001)}]{WuZ01}%
  \BibitemOpen
  \bibfield  {author} {\bibinfo {author} {\bibfnamefont {S.}~\bibnamefont
  {{Wu}}}\ and\ \bibinfo {author} {\bibfnamefont {Y.}~\bibnamefont {{Zhang}}},\
  }\href {\doibase 10.1103/PhysRevA.63.012308} {\bibfield  {journal} {\bibinfo
  {journal} {\pra}\ }\textbf {\bibinfo {volume} {63}},\ \bibinfo {eid} {012308}
  (\bibinfo {year} {2001})},\ \Eprint {http://arxiv.org/abs/quant-ph/0004020}
  {quant-ph/0004020} \BibitemShut {NoStop}%
\bibitem [{\citenamefont {Landauer}(1961)}]{Landauer61}%
  \BibitemOpen
  \bibfield  {author} {\bibinfo {author} {\bibfnamefont {R.}~\bibnamefont
  {Landauer}},\ }\href@noop {} {\bibfield  {journal} {\bibinfo  {journal} {IBM
  J. Res. Dev.}\ }\textbf {\bibinfo {volume} {5}},\ \bibinfo {pages} {183}
  (\bibinfo {year} {1961})}\BibitemShut {NoStop}%
\bibitem [{\citenamefont {Faist}\ \emph
  {et~al.}(2015{\natexlab{b}})\citenamefont {Faist}, \citenamefont
  {Oppenheim},\ and\ \citenamefont {Renner}}]{FaistOR14}%
  \BibitemOpen
  \bibfield  {author} {\bibinfo {author} {\bibfnamefont {P.}~\bibnamefont
  {Faist}}, \bibinfo {author} {\bibfnamefont {J.}~\bibnamefont {Oppenheim}}, \
  and\ \bibinfo {author} {\bibfnamefont {R.}~\bibnamefont {Renner}},\ }\href
  {http://stacks.iop.org/1367-2630/17/i=4/a=043003} {\bibfield  {journal}
  {\bibinfo  {journal} {New Journal of Physics}\ }\textbf {\bibinfo {volume}
  {17}},\ \bibinfo {pages} {043003} (\bibinfo {year}
  {2015}{\natexlab{b}})}\BibitemShut {NoStop}%
\bibitem [{\citenamefont {{Lostaglio}}\ \emph
  {et~al.}(2015{\natexlab{a}})\citenamefont {{Lostaglio}}, \citenamefont
  {{Jennings}},\ and\ \citenamefont {{Rudolph}}}]{LostaglioJR14}%
  \BibitemOpen
  \bibfield  {author} {\bibinfo {author} {\bibfnamefont {M.}~\bibnamefont
  {{Lostaglio}}}, \bibinfo {author} {\bibfnamefont {D.}~\bibnamefont
  {{Jennings}}}, \ and\ \bibinfo {author} {\bibfnamefont {T.}~\bibnamefont
  {{Rudolph}}},\ }\href {\doibase 10.1038/ncomms7383} {\bibfield  {journal}
  {\bibinfo  {journal} {Nature Communications}\ }\textbf {\bibinfo {volume}
  {6}},\ \bibinfo {eid} {6383} (\bibinfo {year}
  {2015}{\natexlab{a}})}\BibitemShut {NoStop}%
\bibitem [{\citenamefont {{Lostaglio}}\ \emph
  {et~al.}(2015{\natexlab{b}})\citenamefont {{Lostaglio}}, \citenamefont
  {{Korzekwa}}, \citenamefont {{Jennings}},\ and\ \citenamefont
  {{Rudolph}}}]{LostaglioKJR15}%
  \BibitemOpen
  \bibfield  {author} {\bibinfo {author} {\bibfnamefont {M.}~\bibnamefont
  {{Lostaglio}}}, \bibinfo {author} {\bibfnamefont {K.}~\bibnamefont
  {{Korzekwa}}}, \bibinfo {author} {\bibfnamefont {D.}~\bibnamefont
  {{Jennings}}}, \ and\ \bibinfo {author} {\bibfnamefont {T.}~\bibnamefont
  {{Rudolph}}},\ }\href {\doibase 10.1103/PhysRevX.5.021001} {\bibfield
  {journal} {\bibinfo  {journal} {Physical Review X}\ }\textbf {\bibinfo
  {volume} {5}},\ \bibinfo {eid} {021001} (\bibinfo {year}
  {2015}{\natexlab{b}})}\BibitemShut {NoStop}%
\bibitem [{\citenamefont {Mossa}\ \emph {et~al.}(2009)\citenamefont {Mossa},
  \citenamefont {Manosas}, \citenamefont {Forns}, \citenamefont {Huguet},\ and\
  \citenamefont {Ritort}}]{MossaMFHR09}%
  \BibitemOpen
  \bibfield  {author} {\bibinfo {author} {\bibfnamefont {A.}~\bibnamefont
  {Mossa}}, \bibinfo {author} {\bibfnamefont {M.}~\bibnamefont {Manosas}},
  \bibinfo {author} {\bibfnamefont {N.}~\bibnamefont {Forns}}, \bibinfo
  {author} {\bibfnamefont {J.~M.}\ \bibnamefont {Huguet}}, \ and\ \bibinfo
  {author} {\bibfnamefont {F.}~\bibnamefont {Ritort}},\ }\href {\doibase
  10.1088/1742-5468/2009/02/P02060} {\bibfield  {journal} {\bibinfo  {journal}
  {Journal of Statistical Mechanics: Theory and Experiment}\ }\textbf {\bibinfo
  {volume} {2009}},\ \bibinfo {pages} {P02060} (\bibinfo {year}
  {2009})}\BibitemShut {NoStop}%
\bibitem [{\citenamefont {{Manosas}}\ \emph {et~al.}(2009)\citenamefont
  {{Manosas}}, \citenamefont {{Mossa}}, \citenamefont {{Forns}}, \citenamefont
  {{Huguet}},\ and\ \citenamefont {{Ritort}}}]{ManosasMFHR09}%
  \BibitemOpen
  \bibfield  {author} {\bibinfo {author} {\bibfnamefont {M.}~\bibnamefont
  {{Manosas}}}, \bibinfo {author} {\bibfnamefont {A.}~\bibnamefont {{Mossa}}},
  \bibinfo {author} {\bibfnamefont {N.}~\bibnamefont {{Forns}}}, \bibinfo
  {author} {\bibfnamefont {J.}~\bibnamefont {{Huguet}}}, \ and\ \bibinfo
  {author} {\bibfnamefont {F.}~\bibnamefont {{Ritort}}},\ }\href {\doibase
  10.1088/1742-5468/2009/02/P02061} {\bibfield  {journal} {\bibinfo  {journal}
  {Journal of Statistical Mechanics: Theory and Experiment}\ }\textbf {\bibinfo
  {volume} {2}},\ \bibinfo {pages} {61} (\bibinfo {year} {2009})}\BibitemShut
  {NoStop}%
\bibitem [{\citenamefont {Renes}\ \emph {et~al.}(2013)\citenamefont {Renes},
  \citenamefont {{Yunger Halpern}},\ and\ \citenamefont
  {{M\"{u}ller}}}]{RenesYHM13}%
  \BibitemOpen
  \bibfield  {author} {\bibinfo {author} {\bibfnamefont {J.~M.}\ \bibnamefont
  {Renes}}, \bibinfo {author} {\bibfnamefont {N.}~\bibnamefont {{Yunger
  Halpern}}}, \ and\ \bibinfo {author} {\bibfnamefont {M.}~\bibnamefont
  {{M\"{u}ller}}},\ }\href@noop {} {}\bibinfo {howpublished} {private
  communication} (\bibinfo {year} {2013})\BibitemShut {NoStop}%
\end{thebibliography}%

\end{document}